\documentclass[11pt]{article}
\usepackage{theapa, rawfonts}

\usepackage{geometry}
\usepackage{macros}
\usepackage{subdepth}
\usepackage{tikz,tkz-graph} 
\usetikzlibrary{positioning}  
\mathchardef\mhyphen="2D 
\usepackage{algorithm,algorithmic}

\usepackage[normalem]{ulem}

\title{On the complexity of finding set repairs for data-graphs} 

\usepackage{authblk}

\author[1,2]{Sergio Abriola}
\author[1]{Santiago Cifuentes}
\author[1,2]{María Vanina Martínez}
\author[1,2]{Nina Pardal}
\author[1]{Edwin Pin}
\affil[1]{University of Buenos Aires, Argentina}
\affil[2]{ICC-CONICET, Argentina}

\begin{document}

\maketitle

\begin{abstract}
In the deeply interconnected world we live in, pieces of information link domains all around us. As graph databases embrace effectively relationships among data and allow processing and querying these connections efficiently, they are rapidly becoming a popular platform for storage that supports a wide range of domains and applications. As in the relational case, it is expected that data preserves a set of integrity constraints that define the semantic structure of the world it represents. When a database does not satisfy its integrity constraints, a possible approach is to search for a `similar' database that does satisfy the constraints, also known as a repair. In this work, we study the problem of computing subset and superset repairs for graph databases with data values using a notion of consistency based on a set of Reg-GXPath expressions as integrity constraints. We show that for positive fragments of Reg-GXPath these problems admit a polynomial-time algorithm, while the full expressive power of the language renders them intractable.
\end{abstract}

\section{Introduction}

The availability of high volumes of interconnected data allows the possibility of developing applications that go well beyond semantic indexing and search and involve advance reasoning tasks on top of existing data. 
Alternative data models such as graph databases are becoming popular as they allow to effectively represent and access this type of data.
Graph databases are specially useful for applications where the topology of the data is as important as the data itself, such as social networks analysis \cite{fan2012graph}, data provenance \cite{anand2010techniques}, and the Semantic Web \cite{arenas2011querying}. 
The structure of the database is commonly queried through navigational languages such as \textit{regular path queries} or RPQs \cite{barcelo2013querying} that can capture pairs of nodes connected by some specific kind of path. These query languages can be extended to add more expressiveness, while usually adding extra complexity in the evaluation as well. For example, C2RPQs  are a natural extension of RPQs defined by adding to the language the capability of traversing edges backwards and closing the expressions under conjunction (similar to relational CQs). 

RPQs and its most common extensions (C2RPQs and NREs \cite{barcelo2012relative}) can only act upon the edges of the graph, leaving behind any possible interaction with data values in the nodes. This led to the design of query languages for the case of data-graphs (i.e. graph databases where data lies both in the paths and in the nodes themselves), such as REMs and \Gregxpath \cite{libkin2016querying}.  

Advanced computational reasoning tasks usually require the management of inconsistent information. Reasoning with or in the presence of inconsistent knowledge bases has been the focus of a vast amount of research in Artificial Intelligence and theory of relational databases for over 30 years. However, in the last years, the area has flourished focusing
on logical knowledge bases and ontologies~\cite{Lembo2010,Bienvenu2019,LukasiewiczMMMP22}

As in the relational case, consistency is related to the notion of \textit{integrity constraints} that express some of the semantic structure the data intends to represent.
In the context of graph databases, these constraints can be expressed in graph databases through \textit{path constraints}~\cite{abiteboul1999regular,buneman2000path}. When a database does not satisfy its integrity constraints, a possible approach is to compute a `similar' database that does satisfy the constraints. In the literature, such a database is called a \emph{repair}~\cite{arenas1999consistent}, and in order to define it properly one has to precisely determine the meaning of `similar'. 

Consider for example the data-graph in Figure~\ref{figure:familia}: here the nodes represent people, and the edges model different kinds of relationships between them, such as brothership or parentship. Observe that in this data-graph, every pair of nodes $(x,y)$ connected with a \esLabel{sibling\_of} edge directed from $x$ to $y$ is also connected with a \esLabel{sibling\_of} edge directed from $y$ to $x$. This property seems reasonable, since the brothership relation is symmetric. In addition, the nodes $(\esDato{Mauro}, \esDato{Julieta})$ are connected through a \esLabel{nibling\_of} edge directed to \esDato{Julieta}, which it is also a property we would expect to find considering that \esDato{Mauro}'s dad, \esDato{Diego}, is the brother of \esDato{Julieta}.

\label{figure:familia}
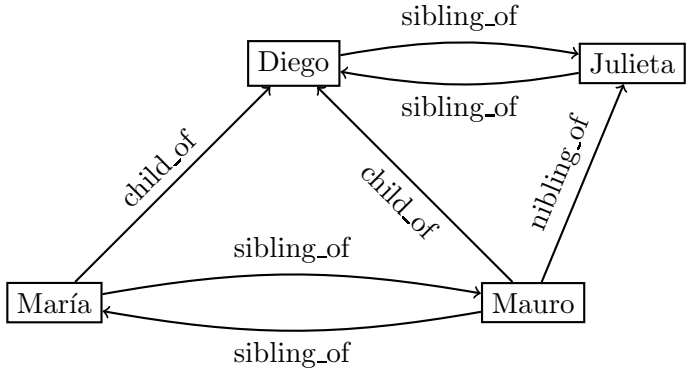
\begin{figure} [H]
\centering
    \begin{tikzpicture}[node distance={45mm}, thick, main/.style = {draw, rectangle}] 
\node[main] (1) {Diego}; 
\node[main] (2) [below right of=1] {Mauro}; 
\node[main] (3) [below left of=1] {María}; 
\node[main] (4) [right of=1] {Julieta}; 
\draw[->] (2) -- (1) node[midway, below=0.5pt, sloped]{child\_of}; 
\draw[->] (3) -- (1) node[midway, above=0.5pt, sloped]{child\_of}; 
\draw[->] (2) -- (4) node[midway, above=0.5pt, sloped, pos=0.5]{nibling\_of}; 
\draw[->, bend left=10] (1) edge node [above=0.3pt]{sibling\_of} (4);
\draw[->] (4) to [out=190, in=350] node [below=0.3pt]{sibling\_of} (1); 

\draw[->, bend left=10] (3) edge node [above=0.3pt]{sibling\_of} (2);
\draw[->] (2) to [out=190, in=350] node [below=0.3pt]{sibling\_of} (3); 
\end{tikzpicture} 
\caption{An example of a data-graph where the nodes represent people and the edges family relationships.}
\end{figure}

The structure that this particular data-graph has to preserve in order to properly capture our notions of \esLabel{sibling\_of} and \esLabel{nibling\_of} can be stated through path expressions that capture those pairs of nodes that satisfy the semantic constraints. In particular, the pair (\esDato{María}, \esDato{Julieta}) does not satisfy the semantics of the \esLabel{nibling\_of} relation, since \esDato{Diego} is also \esDato{María}'s parent but there is no edge \esLabel{nibling\_of} from \esDato{María} to \esDato{Julieta}. This could be `fixed' by adding the \esLabel{nibling\_of} edge from \esDato{Maria} to \esDato{Julieta}, or by rather deleting the edge \esLabel{child\_of} directed from \esDato{María} to \esDato{Diego}. Naturally, we want to preserve as much data as possible from our original data-graph.

There are different notions of repairs in the literature, among others, set-based repairs~\cite{tenCate:2012}, attribute-based repairs~\cite{Wijsen:2003}, and cardinality based repairs~\cite{Lopatenko07complexityof}. 
In this work, we study two restrictions of the problem of finding a set-based repair $G'$ of a graph database $G$ under a set of \Gregxpath expressions $R$ considered as \textit{path constraints}: when $G'$ is a subgraph of $G$ and when $G'$ is a super-graph of $G$. These kinds of repairs are usually called \textit{subset} and \textit{superset} repairs, respectively~\cite{tenCate:2012,barcelo2017data}. Since repairs may not be unique, it is possible to impose an order over the set of repairs and look for an `optimum' repair over such order~\cite{flesca2007preferred,staworko2012prioritized}.

The specific contributions of this work are the following:
\begin{itemize}
    \item We define a model for graph databases with data values and introduce a notion of consistency over this model, based on a set of \Gregxpath expressions that capture a significant group of common integrity constraints that appear in the literature.
    \item Given a graph database $\aGraph$ and a set of constraints $\aRestrictionSet$, we study the problem of computing a subset repair (respectively, a superset repair) $\aGraph'$ of $\aGraph$. 
    \item We show that depending on the expressiveness of the restrictions (whether they use the full power of \Gregxpath or only the fragment without negation known as \Gposregxpath), these problems admit a polynomial-time algorithm or rather turn out to be intractable, even undecidable.
    
\end{itemize}

The rest of this work is organized as follows. In Section~\ref{Section:Definitions} we introduce the necessary preliminaries and notation for the syntax and semantics for our data-graph model as well as the definitions of consistency and different types of repairs.
We also show, by means of examples, that the proposed language
can capture a group of integrity constraints that are common in the database literature. In Section~\ref{Section:Repairs} we study the complexity of the repair computing problem
for both subset and superset. 
Finally, in Section~\ref{Section:RelatedWork} and Section~\ref{Section:Conclusions} we discuss related work and conclusions, respectively.

\section{Definitions}

Fix a finite set of edge labels $\Sigma_e$ and a countable (either finite or infinite enumerable) set of data values $\Sigma_n$, sometimes referred to as data labels, both of which we assume non-empty, and such that $\Sigma_e \cap \Sigma_n = \emptyset$. A \emph{data-graph}~$G$ is a tuple $(V,L_e,\dataFunction)$ where $V$ is a set of nodes, $L_e$ is a mapping from $V \times V$ to $\parts{\Sigma_e}$ defining the edges of the graph, and $\dataFunction$ is a function mapping the nodes from $V$ to data values from $\Sigma_n$. 

\textit{Path expressions} of \Gregxpath are given by:

\begin{center}
    $\aPath, \aPathb$ = $\epsilon$ $|$ $\labelComodin$ $|$ $\aLabel$ $|$ $\aLabel^{-}$ $|$ $\expNodoEnCamino{\aFormula}$ $|$ $\aPath$ . $\aPathb$ $|$ $\aPath \pathUnion \aPathb$ $|$ $\aPath \pathIntersection \aPathb$ $|$ $\aPath^{*}$ $|$ $\pathComplement{\aPath}$ $|$
    $\aPath^{n,m}$ 
\end{center} 

where $\aLabel$ iterates over every label of from $\Sigma_e$ and $\aFormula$ is a \textit{node expression} defined by the following grammar:

\begin{center}
    $\aFormula, \aFormulab$ = $\neg \aFormula \mid \aFormula \wedge \aFormulab$ $|$ $\comparacionCaminos{\aPath}$ $|$ $\esDatoIgual{\aData}$ $|$ $\esDatoDistinto{\aData}$ $|$ $\comparacionCaminos{\aPath = \aPathb}$ $|$ $\comparacionCaminos{\aPath \neq \aPathb}$ $|$
    $\aFormula \vee \aFormulab$
\end{center}

where $\aPath$ and $\aPathb$ are path expressions (i.e. path and node expressions are defined by mutual recursion) and $c$ iterates over $\Sigma_n$. If we only allow the Kleene star to be applied to labels and their inverses (the production $\aLabel^-$), then we obtain a subset of \Gregxpath called \Gcorexpath. The semantics of these languages are defined  in \cite{libkin2016querying} in a similar fashion as the usual regular languages for navigating graphs \cite{barcelo2013querying}, while adding some extra capabilities such as the complement of a path expression $\pathComplement{\aPath}$ and data tests. The $\comparacionCaminos{\aPath}$ operator is the usual one for \textit{nested regular expressions} (NREs) used in \cite{barcelo2012relative}. Given a data-graph $\aGraph=(V,L,D)$, the semantics of \Gregxpath expressions are:

\begin{itemize}[leftmargin=.6in,label={}]

\itemsep0em 

  \item $\semantics{\epsilon}_\aGraph = \{(v,v)$ $|$ $v \in V\}$
  
  \item $\semantics{\labelComodin}_\aGraph = \{(v,w)$ $|$ $v,w \in V, L(v,w) \neq \emptyset\}$ 
  
  \item $\semantics{\aLabel}_\aGraph = \{ (v,w)$ $|$ $\aLabel \in L(v,w)\}$
  
  \item $\semantics{\aLabel^-}_\aGraph = \{ (w,v)$ $|$ $\aLabel \in L(v,w)\}$
  
  \item $\semantics{\aPath^*}_\aGraph = $ the reflexive transitive closure of $\semantics{\aPath}_\aGraph$
  
  \item $\semantics{\aPath . \aPathb}_\aGraph = \semantics{\aPath}_\aGraph  \semantics{\aPathb}_\aGraph$
  
  \item $\semantics{\aPath \pathUnion \aPathb}_\aGraph = \semantics{\aPath}_\aGraph \cup \semantics{\aPathb}_\aGraph$
  
   \item $\semantics{\aPath \pathIntersection \aPathb}_\aGraph = \semantics{\aPath}_\aGraph \cap \semantics{\aPathb}_\aGraph$
    
  \item $\semantics{\pathComplement{\aPath}}_\aGraph = V \times V \setminus \semantics{\aPath}_\aGraph$
  
  \item $\semantics{\expNodoEnCamino{\aFormula}}_\aGraph = \{(v,v)$ $|$ $v \in \semantics{\aFormula}_\aGraph$\}
  
  \item $\semantics{\aPath^{n,m}}_\aGraph = \bigcup\limits_{k=n}^m(\semantics{\aPath}_\aGraph)^k$
  
  \item $\semantics{\comparacionCaminos{\aPath}}_\aGraph = \pi_1(\semantics{\aPath}_\aGraph) = \{v \mid \exists w\in V,\, (v,w) \in \semantics{\aPath}_\aGraph \} $ 
  
  \item $\semantics{\lnot \aFormula}_\aGraph = V \setminus \semantics{\aFormula}_\aGraph$
  
  \item $\semantics{\aFormula \wedge \aFormulab}_\aGraph = \semantics{\aFormula}_\aGraph \cap \semantics{\aFormulab}_\aGraph$
  
  \item $\semantics{\aFormula \vee \aFormulab}_\aGraph = \semantics{\aFormula}_\aGraph \cup \semantics{\aFormulab}_\aGraph$
  
  \item $\semantics{\esDatoIgual{c}}_\aGraph = \{v \in V$ $|$ $D(v) = c$\}
  
  \item $\semantics{\esDatoDistinto{c}}_\aGraph = \{v \in V$ $|$ $D(v) \neq c$\}
  
  \item $\semantics{\comparacionCaminos{\aPath = \aPathb}}_\aGraph = \{v \in V$ $|$ $\exists v', v'' \in V$, $(v,v') \in \semantics{\aPath}_\aGraph$, $(v,v'') \in \semantics{\aPathb}_\aGraph$, $D(v') = D(v'')$\} 
  
  \item $\semantics{\comparacionCaminos{\aPath \neq \aPathb}}_\aGraph = \{v \in V$ $|$ $\exists v', v'' \in V$, $(v,v') \in \semantics{\aPath}_\aGraph$, $(v,v'') \in \semantics{\aPathb}_\aGraph$, $D(v') \neq D(v'')$\} 
  
\end{itemize}

Notice that we are using the standard composition of relations in the definition of $\semantics{\aPath . \aPathb}_\aGraph$ and $\semantics{\aPath^{n,m}}_\aGraph$. More precisely, given two binary relations $R_1,R_2$ over $V_\aGraph$, we define $R_1 R_2$ as $\{(x,z) : \exists y\in  (x, y) \in R_1, (y,z) \in R_2\}$, and $R^{k+1} \equiv R^k R$.

We denote with $\aPath \entoncesCamino \aPathb$ the path expression $\aPathb \pathUnion \pathComplement{\aPath}$, and with $\aNodeExpression \entoncesNodo \aNodeExpressionb$ the node expression $\aNodeExpressionb \lor \neg \aNodeExpression$. We also note a label $\aLabel$ as $\down_\esLabel{\aLabel}$ in order to easily distinguish the `path' fragment from the expressions. For example, the expression \esLabel{child\_of} [\esDatoIgual{Maria}] \esLabel{sister\_of} will be noted as $\down_\esLabel{child\_of} [\esDatoIgual{Maria}] \down_\esLabel{sister\_of}$.

Naturally, the expression $\aPath \cap \aPathb$ can be rewritten as $\overline{\overline{\aPath} \cup \overline{\aPathb}}$ while preserving the semantics, and something similar happens with the operators $\wedge$ and $\vee$ for the case of node expressions using the $\lnot$ operator. However, we still include the definitions of all these operators for this grammar, since we will be interested in the sequel in a fragment of $\Gregxpath$ called $\Gposregxpath$. This fragment shares the same grammar with the exception of the $\overline{\aPath}$ and $\lnot \aFormula$ productions. Thus, we will not be able to `simulate' the $\cap$ operator in \Gposregxpath unless it is present in the original \Gregxpath grammar.

\begin{example}
 Using \Gcorexpath, we could write an expression capturing every `transitive friend` in a social network using  $\aPath = \down_\esLabel{friend\_of}^+$. This is a common example of regular expressions, and it does not require most of \Gcorexpath capacities. We can enhance our example by asking for an expression that captures those pairs of nodes who are `transitive friends` and also follow some other specific node $x$.  
\begin{equation*}
    \aPath = [\comparacionCaminos{\down_\esLabel{follows}[x^=]}]\down_\esLabel{friend\_of}^+ [\comparacionCaminos{\down_\esLabel{follows}[x^=]}]
\end{equation*}
\end{example}

In this case, we are already using the data tests and the \textit{nest} operator. 

Another interesting difference between most regular languages and \Gcorexpath is that we can easily obtain the complement of a path expression. Since regular languages are closed under complement, most navigation languages can express the complement of any expression, but this is not generally built into the grammar.

\begin{example}
 Given a social database where we have family links, we can capture those nodes that do not have any ancestor node with the same name. We do so by using the following node expression:

\begin{center}
    
$\varphi = \lnot \comparacionCaminos{\epsilon = \down_\esLabel{father\_of}^+}$

\end{center}

Note that this cannot  be expressed by most navigational languages, since they do not have a way of comparing data values.
\end{example}

\paragraph{Consistency}
Given a specific database, we may want a node or path expression to capture all the nodes from the data-graph, since it could represent some structure we expect to find in our data. This kind of \Gcorexpath or \Gregxpath expression would work as an \textit{integrity constraint} defining semantic relations among our data. A data-graph in which an expression $\aPath$, used as a constraint, does not capture all nodes, may be called \textit{inconsistent} with respect to $\aPath$. In general, we define the notion of consistency in the following way:

\begin{definition}[Consistency]
Let $\aGraph$ be a data-graph and $\aRestrictionSet = \aRestrictionSetPaths \union \aRestrictionSetNodes$ a finite set of restrictions, where $\aRestrictionSetPaths$ and $\aRestrictionSetNodes$ consist of path and node expressions, respectively. 
We say that $(\aGraph, \aRestrictionSet)$ is \defstyle{consistent} with respect to\ $\aRestrictionSet$, noted as $\aGraph \models \aRestrictionSet$, if the following conditions hold: 
\begin{itemize}

\item $\forall \aNode \in V_\aGraph$ and $\aNodeExpression \in \aRestrictionSetNodes$, we have that $\aNode \in \semantics{\aNodeExpression}$ 
\item $\forall \aNode,\aNodeb \in V_\aGraph$ and $\aPath \in \aRestrictionSetPaths$, we have that $(\aNode,\aNodeb) \in \semantics{\aPath}$
\end {itemize}
Otherwise we say that $\aGraph$ is inconsistent w.r.t.\ $\aRestrictionSet$. 
\end{definition}

In the remainder of this work we will simply say that $\aGraph$ is (in)consistent, whenever $\aRestrictionSet$ is clear from the context. 

\begin{example}

 Consider a film database (see figure \ref{fig:film}), where some nodes represent people from the film industry such as actors or directors, and others represent movies or documentaries. If we want to make a cut from that graph that preserves only actors who have worked with Philip Seymour Hoffman in a film from Paul Thomas Anderson, then we want the following formula to be satisfied:

\begin{center}
    $ \aNodeExpression = \comparacionCaminos{\down_\esLabel{type}[actor^=]}   \entoncesNodo \comparacionCaminos{\down_\esLabel{acts\_in}\comparacionCaminos{\down_\esLabel{directed\_by} [\textit{Anderson}^=]} \down_\esLabel{acts\_in}^- [\textit{Hoffman}^=]}
$
\end{center}

\begin{figure}[H]

\centering

\begin{tikzpicture}[node distance={25mm}, thick, main/.style = {draw, rectangle}]

\node[main] (Hoffman) {Hoffman};
\node[main] (Actor) [left of=Hoffman] {Actor};
\node[main] (Phoenix) [below left of=Actor] {Joaquin Phoenix};
\node[main] (Moore) [below right of=Hoffman] {Julianne Moore};
\node[main] (Magnolia) [above right of=Moore] {Magnolia};
\node[main] (The Master) [above of=Hoffman] {The Master};
\node[main] (film) [right of=Magnolia] {Film};
\node[main] (Anderson) [above of=film] {Anderson};

\draw[->] (Hoffman) -- (Actor) node[midway, above=0.2pt, sloped]{type};

\draw[->] (Phoenix) -- (Actor) node[midway, above=0.1pt, sloped] {type};

\draw[->] (Moore) -- (Actor) node[midway, above=0.2pt, sloped] {type};

\draw[->, bend left=50] (Phoenix) edge  node[midway, above=0.1pt, sloped] {acts\_in} (The Master) ;

\draw[->] (Hoffman) -- (The Master) node[midway, left=2pt]{acts\_in};

\draw[->] (The Master) -- (Anderson) node[midway, above=0.1pt] {directed\_by};

\draw[->] (Magnolia) -- (Anderson) node[midway, above=0.1pt, sloped] {directed\_by};

\draw[->] (Moore) -- (Magnolia) node[midway, above=0.1pt, sloped] {acts\_in};

\draw[->] (Magnolia) -- (film) node[midway, below=0.1pt] {type};

\draw[->, bend right=10] (The Master) edge node[midway, above=0.1pt, sloped] {type} (film) ;

\draw[->] (Hoffman) -- (Magnolia) node[midway, above=0.1pt] {acts\_in};

    
\end{tikzpicture} 

\caption{In this data-graph $\aNodeExpression$ is satisfied, since both Phoenix and Moore have worked with Hoffman in a film from Anderson (respectively through `The Master` and `Magnolia`). Note that the restriction also applies to Hoffman, so it is required that he participates in any film from Anderson in order to satisfy the constraint. }
\label{fig:film}

\end{figure}
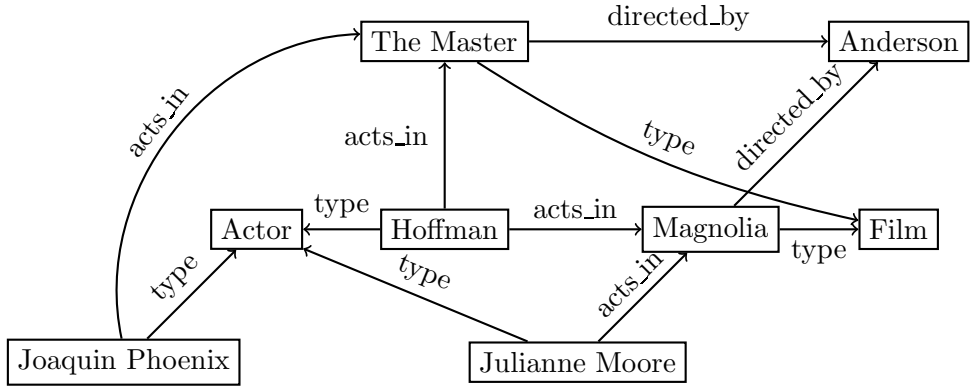
\end{example}

\begin{example} \label{example:FamilyGraph} 
\normalfont All \textit{Regular Path Constraints} (RPCs) considered in \cite{barcelo2017data} can be written as $\Gregxpath$ expressions, since the language is capable to express implications of the form `$\aPath_1 \entoncesCamino \aPath_2$` where $\aPath_i$ is a 2RPQ for $i \in \{1,2\}$. In particular, the constraints from Example 1 of \cite{barcelo2017data} can be written as $\Gcorexpath$ expressions in a straightforward way: the symmetry of the \esLabel{sibling\_of} relation can be stated as $\aPath = \down_\esLabel{sibling\_of} \entoncesCamino$ $ \down_\esLabel{sibling\_of}^{-}$, and the nibling condition as $\aPathb = \down_\esLabel{child\_of} \down_\esLabel{sibling\_of} \entoncesCamino \down_\esLabel{nibling\_of}$. These are exactly the constraints we wanted to express in figure \ref{figure:familia}. Now we can formally state that the data-graph is inconsistent, since (\esDato{María}, \esDato{Julieta}) $\notin \semantics{\aPathb}$. Meanwhile, every pair $(x,y)$ belongs~to~$\semantics{\aPath}$.

Deleting the edge (\esDato{María}, \esLabel{child\_of}, \esDato{Diego}) would make the graph consistent, since \esDato{María} would no longer be \esDato{Julieta}'s nibling and thus $\aPathb$ would be satisfied. On the other hand, deleting the edge (\esDato{Diego}, \esLabel{sibling\_of}, \esDato{Julieta}) would also fix $\aPathb$, but it would cause $\aPath$ to not be satisfied anymore, since the brothership relation would not be symmetric. Note that adding the edge (\esDato{Maria}, \esLabel{nibling\_of}, \esDato{Julieta}) would also make the data-graph consistent.  
\end{example}

\begin{example}
    We can express restrictions that are similar to \textit{foreign keys} using $\Gcorexpath$, relating data values from different entities. Consider a database where we have information about citizens and cities. We may like to ensure that a person's nationality coincides with the country of the city he was born in. This could be captured with the expression:
\begin{equation*}
    \aNodeExpression = \comparacionCaminos{\down_\esLabel{type}[person^=]} \entoncesNodo \comparacionCaminos{\down_\esLabel{bornInCity} \down_\esLabel{nation} = \down_\esLabel{nationality}}
\end{equation*} 
\end{example}

\begin{example}

NREs $\subseteq$ $\Gregxpath$ \cite{libkin2016querying} can express restrictions that take into account typing and inheritance properties common to the \textit{Resource Description Framework} (RDF), such as those considered in \cite{perez2010nsparql}.
\end{example}



\paragraph{Repairs}
When a graph database $\aGraph$ is inconsistent with respect to a set of restrictions $\aRestrictionSet$ (i.e. there is a path expression or node expression in $\aRestrictionSet$ that is not satisfied) we would like to compute a new graph database $\aGraph'$ consistent with respect to $\aRestrictionSet$ and that differs minimally from $\aGraph$. This new database $\aGraph'$ is usually called a \textit{repair} of $\aGraph$ with respect to $\aRestrictionSet$ following some formal definition for the semantics of `minimal difference'. 

%

We consider \textit{set repairs}, in which the notion of minimality is based on the difference between the sets of nodes and edges. While we could provide a notion of distance between arbitrary data-graphs via an adequate definition of symmetric difference, it has been the case that the complexity of finding such repairs is quite high. Thus, it is common to consider only those set repairs in which one graph is obtained from the other by only adding or only deleting information~\cite{barcelo2017data,lukasiewicz2013complexity,tenCate:2012}. This gives raise to subset and superset repairs.

We say that a data-graph $\aGraph = (V,L_e,D)$ is a \defstyle{subset} of a data-graph $\aGraph'=(V',L_e',D')$ (noted as $\aGraph \subseteq \aGraph'$) if and only if $V \subseteq V'$, and for every pair $v,v' \in V$ holds that $L_e(v,v') \subseteq L_e'(v,v')$ and $D(v) = D'(v)$. In this case, we also say that $\aGraph'$ is a \defstyle{superset} of $\aGraph$.




    
    
        
        



\begin{definition}[Subset  and superset repairs]

Let $\aRestrictionSet$ be a set of restrictions and $\aGraph$ a data-graph. We say that $\aGraph'$ is a \defstyle{subset repair} (respectively, \defstyle{superset repair}) or $\subseteq$-repair (respectively, $\supseteq$-repair) of $\aGraph$ if: 

\begin{itemize}
    \item $(\aGraph',\aRestrictionSet)$ is \defstyle{consistent} (i.e. $\aGraph' \models \aRestrictionSet$) 
    
    \item $\aGraph' \subseteq \aGraph$ (respectively, $\aGraph' \supseteq \aGraph$).
        
   \item There is no data-graph $\aGraph''$ such that $(\aGraph'',\aRestrictionSet)$ is consistent and $ \aGraph' \subset \aGraph'' \subseteq \aGraph$ (respectively, $\aGraph' \supset \aGraph'' \supseteq \aGraph$)   
\end{itemize}

We denote the set of subset repairs of $\aGraph$ with respect to $\aRestrictionSet$ as $\subseteq$-$Rep(\aGraph,\aRestrictionSet)$ (respectively, the set of superset repairs as $\supseteq$-$Rep(\aGraph,\aRestrictionSet)$). 

\end{definition}

\begin{example}
 In Example~\ref{figure:familia}, deleting the edge (\esDato{Maria}, \esLabel{child\_of}, \esDato{Diego}) results in a $\subseteq$-repair, while adding the edge (\esDato{Maria}, \esLabel{nibling\_of}, \esDato{Diego}) results in a $\supseteq$-repair. Deleting both (\esDato{Diego}, \esLabel{sibling\_of}, \esDato{Julieta}) and (\esDato{Julieta}, \esLabel{sibling\_of}, \esDato{Diego}) also creates a $\subseteq$-repair, since this new graph database is consistent with respect to $R$ and is also maximal.
\end{example}

\label{Section:Definitions}

\section{Computing repairs}\label{Section:Repairs}




In this Section, we study the computational complexity of the repair computing problem for a given data-graph and a set of constraints.
We start by mentioning a useful fact from~\cite{libkin2016querying} related to \Gregxpath expressions that we use later in this section.

\begin{theorem}\label{teo:PGxpath}
Given a \Gregxpath expression $\alpha$ and a data-graph $\aGraph$ it is possible to compute the set $\semantics{\alpha}_\aGraph$ in polynomial time in the size of $\aGraph$ and $\aPath$.
\end{theorem}

It follows from this result that, given a set of \Gregxpath expressions $\aRestrictionSet$ where $|R| = \Sigma_{\alpha \in \aRestrictionSet} |\alpha|$ and a data-graph $\aGraph$, it is possible to check if $\aGraph$ is consistent with respect to $\aRestrictionSet$ in polynomial time.

In the sequel, we study the complexity of the subset and superset repair-computing problems, which naturally depend on the fragment of \Gregxpath used to express the constraints.
For some results we consider $\aRestrictionSet$ fixed, which is a common simplification of the repair-computing problem usually called the \textbf{data complexity} of the problem.
We obtain bounds for both the combined and data complexity of the problems~\cite{vardi1982complexity}.

\subsection{Subset repairs}

Let us consider the \emph{empty graph} defined as $\aGraph = (\emptyset, L, D)$, from now on denoted by $\emptyset$. Since the empty graph satisfies every set of restrictions, we conclude that every graph $\aGraph$ has a subset repair given any set $\aRestrictionSet$ of restrictions. In order to understand the complexity of finding such repairs, we define the following decision problem:

\begin{center}

\fbox{\begin{minipage}{30em}
  \textsc{Problem}: \textsc{$\exists$SUBSET-REPAIR}

\textsc{Input}: A data-graph $\aGraph$ and a set $\aRestrictionSet$ of expressions from $\mathcal{L}$.

\textsc{Output}: Decide if $\aGraph$ has a subset repair $\aGraph' \neq \emptyset$ with respect to $\aRestrictionSet$.
\end{minipage}}
\end{center}

The difficulty of this problem depends on the set of expressions $\mathcal{L}$. Throughout this work, we will assume that $\mathcal{L}$ is always a set of expressions included in \Gregxpath. Observe that this problem can be reduced to the problem of finding a subset repair of $\aGraph$ with respect to $\aRestrictionSet$ and, therefore, by studying it we can derive lower bounds for the problem of computing subset repairs.


 \textsc{$\exists$SUBSET-REPAIR$_\mathcal{L}$} is in \textsc{NP}, since we can ask a non-empty subset of $\aGraph$ that satisfies $\aRestrictionSet$ as a positive certificate. In what follows, we prove that this problem is NP-complete even if considering only path expressions from $\Gposregxpath$.

\begin{theorem}\label{teo:subsetPositiveIntractable}

There exists a set \aRestrictionSet of \Gposregxpath path expressions such that the problem $\exists$\textsc{SUBSET-REPAIR} is \textsc{NP-complete}.

\end{theorem}

\begin{proof}
We reduce \textsc{3-SAT} to $\exists$\textsc{SUBSET-REPAIR}, considering a fixed set $\aRestrictionSet$ of restrictions. Given a 3CNF formula $\phi$ of $n$ variables $x_1,\ldots,x_n$ and $m$ clauses $c_1,\ldots,c_m$ we want to construct a data-graph $\aGraph=(V,L,D)$ of \Gposregxpath expressions such that $\aGraph$ has a non-trivial subset repair with respect to $\aRestrictionSet$ if and only if $\phi$ is satisfiable. We give a description of $\aRestrictionSet$ once we outline the structure of the data-graph $\aGraph$.

First, let us define the nodes of the graph. We have two `boolean nodes' for each variable, representing each possible assignment. In addition, we have one node for every clause.

\begin{center}
    $V = \{$ $\bot_i$ $|$ $1\leq i \leq n$ $\}$ $\cup$ $\{$ $\top_i$ $|$ $1\leq i \leq n$ $\}$ $\cup$ $\{$ $c_j$ $|$ $1 \leq j \leq m$ $\}$
\end{center}

In a valid subset repair, only one of the nodes $\bot_i$ or $\top_i$ will remain, and that will implicitly define a truth assignment on the variable $x_i$ of $\phi$.

We encode the information of the formula $\phi$ into the edges of our data-graph. For every clause node $c_j$ there is an edge to a boolean node $\bot_i$ (respectively, $\top_i$) if and only if the literal $\lnot x_i$ appears in $c_j$ (respectively without $\lnot$). For this, we consider $\Sigma_e \supseteq \{\esLabel{needs}, \esLabel{exists}, \esLabel{unique}\}$ and define:

\begin{align*}
L(c_j,\top_i) &= \{\esLabel{needs}\} \mbox{ if } x_i \mbox{ appears in } c_j \\
L(c_j,\bot_i) &= \{\esLabel{needs}\} \mbox{ if } \lnot x_i \mbox{ appears in }c_j
\end{align*}

In a valid subset repair, we want every clause $c_j$ to keep at least one of these edges, since that would imply that $c_j$ is satisfied.

However, a technical problem arises from this approach: it could be the case that not every node is present in a repair, and thus this might result in the corresponding clause node $c_j$ to be removed. To avoid this scenario, we add some edges to our graph and a path expression to $\aRestrictionSet$ that will force every node to be present in case the repair is non-trivial. To do this, we define a cyclic structure with $\esLabel{exist}$ edges:

\begin{center}

    $L(c_j,c_{j+1}) = \{\esLabel{exists}\}$
    
    We add $\esLabel{exists}$ edges to $L(c_m,\top_1)$, $L(c_m,\bot_1)$, $L(\top_n,c_1)$ and $L(\bot_n,c_1)$
    
    $L((\star^1)_i, (\star^2)_{i+1}) = \{\esLabel{exists}\}$ for $\star^1, \star^2 \in \{\top, \bot\}$

\end{center}

Using these edges and the following constraint:

\begin{equation} \label{eq:exists}
    \aPathb_1 = \down_\esLabel{exists}^+
\end{equation}

we ensure that every clause node and at least one $\star_i$ node is present for $\star \in \{\bot, \top\}, 1\leq i \leq n$.

We also add $\esLabel{unique}$ edges to every pair of nodes but those of the form $(\top_i, \bot_i)$. In other words, $\esLabel{unique} \in L(v,v')$ for every $v,v'\in V$ such that either $v \neq \top_i$ or $v' \neq \bot_i$ for every $i$. Observe that every pair of nodes except for those of the form $(\top_i, \bot_i)$ belongs to the \esLabel{unique} relation, and if $\aRestrictionSet$ contains the constraint:

\begin{equation} \label{eq:unique}
    \aPathb_2 = \down_\esLabel{unique},
\end{equation}

then every non-trivial repair will not contain both $\top_i$ and $\bot_i$ nodes for each $i$. Hence, every non-trivial repair will implicitly define an assignment on the variables of $\phi$.

Finally, we need to ensure that the assignment defined by any non-trivial repair satisfies $\phi$. To do this, we add more edges: we add a $\esLabel{valid}$ edge between every pair of nodes $v$ and $v'$ such that $(v,v')\neq (c_j,\star_i)$ for every $j$, $i$ and $\star \in \{\bot, \top\}$; and consider the following constraint:

\begin{equation} \label{eq:valid}
    \aPathb_3 = \down_\esLabel{valid} \cup \down_\esLabel{needs} \down_\esLabel{valid}
\end{equation}

Every pair of nodes belongs to the \esLabel{valid} relation, except for those of the form $(c_j,\star_i)$, where $\star \in \{\bot, \top\}$. Moreover, those pairs of nodes will satisfy path expression \ref{eq:valid} only if at least one of the $\esLabel{needs}$ edges remains for every $c_j$. This means that in a non-trivial repair, every clause $c_j$ must have a $\esLabel{needs}$ edge directed to one of its variables valuations that evaluates $c_j$ to $\top$. Thus, the valuation implicitly defined by the $\bot_i$,$\top_i$ nodes that remain in the repair will be a satisfying one.

Intuitively, if there is a non-trivial repair of $\aGraph$ with respect to $\aRestrictionSet = \{\aPathb_1, \aPathb_2, \aPathb_3\}$, then such a repair consists of a selection of boolean nodes for each $i$ such that every clause node keeps a \esLabel{needs} edge.  Notice that $\aRestrictionSet$ is independent of $\phi$.

Let us prove that $\phi$ is satisfiable if and only if $\aGraph$ has a non-trivial repair with respect to $\aRestrictionSet$.

$\implies$) If $\phi$ is satisfiable then there is an assignment $f$ of its variables such that $\phi$ evaluates to true. Consider the graph $\aGraph$, and delete from $\aGraph$ every node $\star_i$ for $\star \in \{\bot, \top\}$ such that $f(x_i) \neq \star_i$.

Since we deleted only one boolean node for each $i$, the path expression~\ref{eq:exists} is satisfied. The expression~\ref{eq:unique} is satisfied, since there is no $i$ for which both $\top_i$ and $\bot_i$ remain in the graph. Finally, since $f$ is a valid assignment, we know that for every clause node $c_j$ one of its $\esLabel{needs}$ edges remains, hence the expression~\ref{eq:valid} is satisfied. 

This implies that the subgraph satisfies $\aRestrictionSet$, which means that there exists a non-trivial subset repair of $\aGraph$ with respect to $\aRestrictionSet$.

$\impliedby$) Let $\aGraph'$ be a non-trivial repair of $\aGraph$ with respect to $\aRestrictionSet$. Since both expressions \ref{eq:exists} and \ref{eq:unique} are satisfied, we know that, for every $i$, one of the boolean nodes belongs to $\aGraph'$. Hence, we define an assignment $f$ on the variables $x_i$ as $f(x_i) = \top$ $\iff$ $\top_i \in V_{\aGraph'}$. Since the expression \ref{eq:valid} is satisfied, at least one variable from each clause evaluates to true (or false, if the variable is negated in the clause). It follows that $f$ is an assignment that evaluates $\phi$ to true.
\end{proof}

The problem is also \textsc{NP-hard} when considering only node expressions from \Gregxpath:

\begin{theorem}\label{teo:subsetNodeHardness}
The problem $\exists$\textsc{SUBSET-REPAIR} is NP-complete for a fixed set of \Gregxpath node expressions.
\end{theorem}

\begin{proof}

We reduce \textsc{3-SAT} to $\exists$\textsc{SUBSET-REPAIR}, considering a fixed set $\aRestrictionSet$ of node expressions from $\Gregxpath$. Given a 3CNF formula $\phi$ of $n$ variables $x_1,\ldots,x_n$ and $m$ clauses $c_1,\ldots,c_m$, we will construct a data-graph $\aGraph=(V,L,D)$ of \Gregxpath expressions such that $\aGraph$ has a non-trivial subset repair with respect to $\aRestrictionSet$ if and only if $\phi$ is satisfiable. We give a description of $\aRestrictionSet$ once we outline the structure of the data-graph $\aGraph$. For our reduction we require that $|\Sigma_e| \geq 4 $ and $|\Sigma_n| \geq 4$, and in particular we will assume that we have edge labels $\{\esLabel{assign},\esLabel{needs\_true},\esLabel{needs\_false},\esLabel{h}\} \subseteq \Sigma_e$ and data values $\{var, clause, \top, \bot\} \subseteq \Sigma_n$.

The data-graph $G=(V,L,D)$ is defined as

\begin{align*}
V =& \{\bot, \top\} \cup \{x_i : 1\leq i \leq n\} \cup \{c_j : 1\leq j \leq m\} \\
L(x_i,\star) =& \{\down_{\esLabel{assign}}\} \mbox{ for } i=1\ldots n, \star \in \{\bot, \top\} \\
L(c_j, x_i) =& \{\down_\esLabel{needs\_true}\} \text{ if $x_i$ is a literal in $c_j$ and $\lnot x_i$ is not} \\
L(c_j, x_i) =& \{\down_\esLabel{needs\_false}\} \text{ if $\lnot x_i$ is a literal in $c_j$ and $x_i$ is not} \\
L(c_j, x_i) =& \{\down_\esLabel{needs\_true},\down_\esLabel{needs\_false}\} \text{ if $x_i$ and $\lnot x_i$ are both literals in $c_j$} \\
L(v,w) =& \emptyset \text{ otherwise} \\
D(x_i) =&\, var \text{ for } i=1\ldots n \\
D(c_j) =&\, clause \text{ for } j=1\ldots m \\
D(\star) =& \star \text{ for } \star \in \{\bot, \top\}
\end{align*}
Given an ordering $p_1 \ldots p_{|V|}$ of the $|V|$ nodes of the data-graph, we add the edge labels $\down_\esLabel{h}$ to $L(p_k,p_{k+1 \mod |V|})$ for $k = 1,\ldots, |V|$.

The set of node expressions $R=\{\psi_1,\psi_2,\psi_3\}$ is defined as:

\begin{align*}
\psi_1 = &\comparacionCaminos{\down_{h}} \\ 
\psi_2 = &\expNodoEnCamino{var^{\neq} \vee \lnot \comparacionCaminos{\down_\esLabel{assign} \neq \down_\esLabel{assign}}}\\ 
\psi_3 = &\expNodoEnCamino{clause^{\neq} \vee
\comparacionCaminos{\down_\esLabel{needs\_true} \down_\esLabel{assign} \expNodoEnCamino{\top}} \vee
\comparacionCaminos{
\down_\esLabel{needs\_false} \down_\esLabel{assign} \expNodoEnCamino{\bot} 
}}
\end{align*}

The data-graph $\aGraph$ is a representation of the input formula $\phi$, where each variable is assigned to both boolean values $\bot$ and $\top$. A subset repair of $\aGraph$ with respect to $\aRestrictionSet$ will represent a proper assignment, where each variable node $x_i$ has a unique outgoing edge of label $\esLabel{assign}$. This condition is imposed by $\psi_2$, since it forbids the case in which a node with data value $var$ has two outgoing edges to nodes having different data value.

Observe that $\psi_3$ implies that in a valid repair every clause will either contain a literal that is a negation of a variable that is assigned to $\bot$, or a non-negated literal assigned to $\top$. Finally, $\psi_1$ forces the repair to either preserve all nodes or none of them: if a proper subset of the nodes is deleted then at least one node will not have an outgoing edge of label $\esLabel{h}$, and therefore $R$ will not be satisfied.

Now, we prove that $\phi$ is satisfiable if and only if $\aGraph$ has a non-trivial subset repair with respect to $\aRestrictionSet$.

$\implies$) Let $f:\{x_i\}_{1\leq i\leq n} \to \{\top,\bot\}$ be an assignment that satisfies $\phi$. Hence, we can define a sub data-graph $G'$ of $\aGraph$ by removing the edges $(x_i,\esLabel{assign},\lnot f(x_i))$. Since every node is present and every $var$ node has an unique outgoing edge of label \esLabel{assign}, $\psi_1$ and $\psi_2$ are satisfied in this data-graph. Moreover, $\psi_3$ is also satisfied: given a $clause$ node $c_j$, there is a literal $x_i$ or $\lnot x_i$ in $c_j$ that is satisfied through $f$, and the data-graph $G'$ contains the edges of the form $(x_i,\esLabel{assign},f(x_i))$. Therefore, from every $clause$ node there will be a path with the proper edge labels to either $\bot$ or $\top$. Since $G'$ satisfies $\aRestrictionSet$ and $G' \subset G$, there must be a non-trivial subset repair $G'' \supseteq G'$ of $G$ with respect to $\aRestrictionSet$.

$\impliedby$) Given a non-trivial subset repair $G'$ of $G$ with respect to $\aRestrictionSet$, we build a satisfying assignment of $\phi$. Since $\psi_1$ is satisfied by $\aGraph'$, every node from $G$ is present in $G'$. We define the following assignment $f:\{x_i\}_{1\leq i\leq n} \to \{\top,\bot\}$ such that $f(x_i) = \top \iff \esLabel{assign} \in L_{G'}(x_i,\top)$. We claim that this assignment satisfies $\phi$.

Given a clause $c_j$ from $\phi$, the node $c_j$ in $G'$ satisfies $\psi_3$. This means that, there is either a path $\down_\esLabel{needs\_true} \down_\esLabel{assign} [\top]$ or a path $\down_\esLabel{needs\_false} \down_\esLabel{assign} [\bot]$ starting at the node $c_j$. Assume without loss of generality that the path is of the form $\down_\esLabel{needs\_false} \down_\esLabel{assign} [\bot]$ (the other case follows analogously), and that the intermediate node is $x_i$. 
By construction, we know that $\lnot x_i$ is a literal in $c_j$. Due to constraint $\psi_2$, the node $x_i$ has at most one outgoing edge of label \esLabel{assign}, and therefore $\esLabel{assign} \in L(x_i,\bot)$ implies that $\esLabel{assign} \notin L(x_i,\top)$, which by definition of $f$ also implies that $f(x_i)= \bot$. We conclude that $c_j$ is satisfied.

\end{proof}

Since computing a repair seems unfeasible when using path constraints from $\Gposregxpath$ or node constraints from $\Gregxpath$, we now study the $\exists$\textsc{SUBSET-REPAIR} problem when $\mathcal{L}$ contains only \textit{node expressions} from $\Gposregxpath$.

First, we prove that the positive fragment of $\Gposregxpath$ satisfies a certain property of monotony:

\begin{lemma} [Monotony of \Gposregxpath] \label{monotony}
    Let $\aGraph$ be a data-graph, $\aPath$ a \Gposregxpath path expression, $\aFormula$ a \Gposregxpath node expression, and $\aGraph'$ a data-graph such that $\aGraph \subseteq \aGraph'$. Then:
    
    \begin{itemize}
    
        \item $\semantics{\aPath}_{\aGraph} \subseteq \semantics{\aPath}_{\aGraph'}$.
        
        \item $\semantics{\aFormula}_{\aGraph} \subseteq \semantics{\aFormula}_{\aGraph'}$.
    
    \end{itemize}
    
\end{lemma}
\begin{proof} 
Intuitively, if a pair of nodes satisfies $(v,w) \in \semantics{\aPath}_\aGraph$ where $\aPath$ is a positive expression, then there is a certain subgraph $H$ of $\aGraph$ that is a witness of that fact. For example, if $(v,w) \in \semantics{\down_l [c^=] \down_l}_\aGraph$ then there is a node $z$ such that $D(z) = c$ and a path $vzw$ made by edges with label $l$. The subgraph composed by these three nodes and two edges is enough for the pair $(v,w)$  to satisfy the positive path expression $\aPath$, and therefore in every superset $G'$ of $G$ it will be the case that $(v,w) \in \semantics{\aPath}_G$. See the appendix~\ref{Appendix} for the proper proof.
\end{proof}

This fact allows us to define an efficient procedure for finding a subset repair, based on the following observation: if a node $v$ in $\aGraph$ does not satisfy a positive node expression $\aFormula$, then there is no subset repair $\aGraph'$ such that $v$ is a node from $\aGraph'$.
More generally, we have the following result.

\begin{theorem}\label{procedureCorrectness}

    Let $\aGraph$ be a data-graph, $\aRestrictionSet$ a set of \Gposregxpath restrictions and $v$ a node from $\aGraph$ which violates a node expression. Then
    
\begin{equation*}
    \subseteq \mhyphen Rep(\aGraph,\aRestrictionSet) = \subseteq \mhyphen Rep(\aGraph_{V_\aGraph \setminus \{v\}},\aRestrictionSet)
\end{equation*}

where $\aGraph_{V}$ is the sub data-graph of $\aGraph$ induced by the nodes $V\subseteq V_\aGraph$
        
\end{theorem}

\begin{proof}
For the $\subseteq$ direction: let $H \in Rep(\aGraph,\aRestrictionSet)$ be a subset repair of $\aGraph$ with respect to $\aRestrictionSet$. Due to Lemma \ref{monotony} and the fact that there is a node expression $\aFormula \in \aRestrictionSet$ such that $v \notin \semantics{\aFormula}_{\aGraph}$ we conclude that $v$ is not a node from $H$. This implies that $H \subseteq G_{V_\aGraph \setminus \{v\}}$, and since $H$ is a maximal consistent subset of $G$ with respect to $\aRestrictionSet$, it also is a maximal consistent subset from $G_{V_\aGraph \setminus \{v\}}$.

For the other direction: let $H \in Rep(G_{V_\aGraph \setminus \{v\}}, R)$. Since $H$ is consistent with respect to $\aRestrictionSet$ and $H \subseteq \aGraph$ we only need to prove that it is maximal with respect to $\aGraph$. Toward a contradiction, suppose there exists a data-graph $H'$ such that $H \subset H' \subseteq G$ and $H'$ satisfies $\aRestrictionSet$. Since $v$ is not a node from $H'$, it follows that $H' \subseteq G_{V_\aGraph \setminus \{v\}}$, which then implies that $H$ is not a repair of $G_{V_\aGraph \setminus \{v\}}$. This results in a contradiction, and hence we proved that $H$ is a subset repair of $\aGraph$ with respect to $\aRestrictionSet$.
\end{proof}

Furthermore, given two data-graphs $\aGraph_1$ and $\aGraph_2$ satisfying a \Gposregxpath node expression $\aFormula$ it can be shown that $\aGraph_1 \cup \aGraph_2$ satisfies $\aFormula$ as well  (this follows from Lemma~\ref{monotony}). Then, if $\aRestrictionSet$ only contains \Gposregxpath node expressions we can conclude that there is a unique subset repair of $\aGraph$ with respect to $\aRestrictionSet$.

Given all these facts, we define an algorithm that computes the unique subset repair of a data-graph given a set of \Gposregxpath node expressions:

\begin{algorithm}[H]

\begin{algorithmic}[1]

\REQUIRE $\aGraph$ is a data-graph and $\aRestrictionSet$ a set of $\Gposregxpath$ node expressions.

\WHILE{$(\aGraph,\aRestrictionSet)$ is inconsistent}
    \STATE $V_{\bot} \leftarrow \{$ $v$ $|$ $v \in V_{\aGraph}$ and $\exists$ $\aFormula \in \aRestrictionSet$ such that $v \notin \semantics{\aFormula}_\aGraph$\}
    \STATE $\aGraph \leftarrow \aGraph_{V_{\aGraph} \setminus V_{\bot}}$
\ENDWHILE
\RETURN $\aGraph$
\end{algorithmic}
\caption{$SubsetRepair(\aGraph,\aRestrictionSet)$}
\label{algorithm:subset}
\end{algorithm}

This method is correct since Theorem~\ref{procedureCorrectness} implies that $Rep(\aGraph,\aRestrictionSet) = Rep(\aGraph_{V_{\aGraph}\setminus V_{\bot}}, \aRestrictionSet)$. It also terminates, since $(\emptyset,\aRestrictionSet)$ is consistent. The set $V_{\bot}$ can be computed in polynomial time, and since there are at most $|V_{\aGraph}|$ iterations we conclude the following:

\begin{theorem}\label{teo:uniqueSubsetRepair}
    Given a data-graph $\aGraph$ and a set of \Gposregxpath node expressions $\aRestrictionSet$, it is possible to compute the unique subset repair of $\aGraph$ with respect to $\aRestrictionSet$ in polynomial time. 
\end{theorem}

Note that Theorem~\ref{teo:subsetPositiveIntractable} readily implies that the problem $\exists$\textsc{SUBSET-REPAIR} is $\textsc{NP-hard}$ for $\mathcal{L} = \Gposregxpath$ if we allow both node and path expressions, even when considering the data complexity of the problem.

We conclude this section by noticing that even though the problem of finding subset repairs turns out to be \textsc{NP-complete} for a fragment of quite simple path expressions, node expressions on the other hand have substantial expressive power, as we show in our examples. 
These examples actually use node expression negation for building the `implications', but they can be rewritten to avoid this by taking into account some extra assumptions on the database.

The examples use node expression negation to make `typed restrictions' over the data-graph. It would be reasonable to assume that a constraint of the form $\comparacionCaminos{\down_\esLabel{type}}$ will be present in the database, and we could exploit this one to modify the other restrictions: instead of writing $type(\alpha) \implies restriction$ we could define $\bigvee_{\beta \neq \alpha}type(\beta) \cup restriction$ (supposing that $\Sigma_n$ is finite), which, given the previous `type' constraint, will retain the original semantics of the implication. Also, since the set of node labels is fixed, this will only increment the expression length by a constant factor. 

Finally, we remark that the algorithm for $\Gposregxpath$ node expressions will work as long as the expressions satisfy monotony. This means that we could add more monotone tools to the language and the procedure would still be correct and run in polynomial type, assuming of course that the expressions from the new language can be evaluated in polynomial time given the length of the expressions and the data-graph.

\subsection{Superset repairs}

We start this section by noticing that, in contrast with the subset repair problem, not every data-graph $\aGraph$ admits a superset repair with respect to a set of \Gregxpath expressions:

\begin{remark}\label{remark:20}

    There exists a data-graph $\aGraph$ and a set $\aRestrictionSet$ of \Gposregxpath expressions such that there is no superset repair of $\aGraph$ with respect to $\aRestrictionSet$.

\end{remark}

For example, consider $\aGraph = (\{v\}, L, D)$ where $D(v) = c$, $L(v,v)=\emptyset$ and the set $\aRestrictionSet$ with only one node expression $\phi = [c^{\neq}]$. Every superset $\aGraph'$ of $\aGraph$ contains $v$ with data value $c$, which implies that $v \notin \semantics{\phi}_{\aGraph'}$, and thus $\aGraph'$ is not consistent with respect to $\aRestrictionSet$. 

In order to study the complexity of finding superset repairs, we define the following decision problem:

\begin{center}

\fbox{\begin{minipage}{30em}
  \textsc{Problem}: \textsc{$\exists$SUPERSET-REPAIR}

\textsc{Input}: A data-graph $\aGraph$ and a set $\aRestrictionSet$ of $\mathcal{L}$ expressions.

\textsc{Output}: Does $\aGraph$ have a superset repair with respect to $\aRestrictionSet$?
\end{minipage}}

\end{center}

When fixing \aRestrictionSet the problem remains intractable in general. 
\begin{theorem}\label{teo:supersetIndecidible}
There is a set of labels $\Sigma_e$ such that 
\decisionProblem{$\exists$SUPERSET-REPAIR} is undecidable for a fixed set $\aRestrictionSet$ of path expressions from \GRegXPath.
\end{theorem}

\begin{proof}
To prove undecidability, we reduce the superset-CQA problem \cite [Theorem 4]{barcelo2017data} to our problem. Since the problem is undecidable, so will be ours.

In \cite [Theorem 4]{barcelo2017data}, the following \decisionProblem{CQA problem} is proven to be undecidable for a particular choice of $\Sigma, q, \Gamma$: given a finite alphabet $\Sigma$, a non-recursive RPQ query $q$, a set of word constraints $\Gamma$, a graph $\aGraph$, and a tuple $(\aNode,\aNodeb)$ of nodes of $\aGraph$, decide whether $(\aNode,\aNodeb)\in \semantics{q}_{\aGraph'}$ for every $\aGraph' \in \supseteq$-$Rep(\aGraph, \Gamma)$ (i.e., there is a $q$-labeled path from $\aNode$ to $\aNodeb$ in all supersets of $\aGraph$ satisfying the constraints $\Gamma$).


For ease of reference, we now provide the required definitions.
Given an alphabet $\Sigma$, regular path queries over $\Sigma$, noted \RPQ, are defined by the following grammar: 
\begin{equation} \label{eq:rpq}
    \anRPQ = \epsilon \mid \aLabel 
    \mid \anRPQ.\anRPQ  \mid \anRPQ \cup \anRPQ \mid \anRPQ^* 
\end{equation} 

We say that an RPQ is \emph{non-recursive} if it does not mention the Kleene-star.

The semantics for \RPQ formulas over a graph $\aGraph$ with edges labeled in $\Sigma$ is as follows:

\begin{itemize}[leftmargin=.6in,label={}]
\itemsep0em 
  \item $\semantics{\epsilon}_\aGraph = \{(v,v) \mid v \in V\}$
  \item $\semantics{\aLabel}_\aGraph = \{(v,w) \mid v,w \in V, \aLabel \in L(v,w)\}$  
  \item $\semantics{\anRPQ.\anRPQ'}_\aGraph = \{(v,x) \mid \exists w \in \aGraph \mbox{ s.t. } (v,w) \in \semantics{\anRPQ}_\aGraph \mbox{ and } (w,x) \in \semantics{\anRPQ'}_\aGraph\}$
  \item $\semantics{\anRPQ \cup \anRPQ'}_\aGraph = \semantics{\anRPQ}_\aGraph \cup \semantics{\anRPQ'}_\aGraph$
  \item $\semantics{\anRPQ^{*}}_\aGraph = \{ (v,w) \mid (v,w) \mbox{ belongs to the reflexive-transitive closure of $\semantics{\anRPQ}_\aGraph$}\}$
\end{itemize}

A word constraint is defined as a formula $\aPath_1 \subseteq \aPath_2$ where $\aPath_i$ is a word formula for $i \in \{1,2\}$; that is, a finite conjunction of labels. Intuitively, a graph $\aGraph$ satisfies a word constraint $\aPath_1 \subseteq \aPath_2$ if $\semantics{\aPath_1}_\aGraph \subseteq \semantics{\aPath_2}_\aGraph$. 

We note that a word constraint $\alpha = \aLabel_1 \dots \aLabel_m \subseteq \aLabel'_1 \dots \aLabel'_n$ is equivalent to a $\GRegXPath$ formula $\translationRPQsToGXPath{\aPath}=\down_{\aLabel_1} \dots \down_{\aLabel_m} \entoncesCamino \down_{\aLabel'_1} \dots \down_{\aLabel'_n}$ (that is, $\down_{\aLabel'_1} \dots \down_{\aLabel'_n} \pathUnion\pathComplement{\down_{\aLabel_1} \dots \down_{\aLabel_m}}$). For a non-recursive RPQ $q$ there is also a straightforward translation to an equivalent $\GRegXPath$ formula $\translationRPQsToGXPath{q}$ that preserves its semantics, since a non-recursive RPQ is a finite union of word formulas.

We define $\aRestrictionSet' = (\bigcup\limits_{\aPath \in \Gamma} \{\translationRPQsToGXPath{\aPath}\}) \cup \{\aPathb \entoncesCamino \pathComplement{\translationRPQsToGXPath{q}}\}$, where $\aPathb$ is the path expression $\downarrow_\esLabel{x}$. We consider the set of edge labels $\Sigma_e = \{\down_{\aLabel} \mid \aLabel \in \Sigma\} \sqcup \{\downarrow_{\esLabel{x}}\}$. Given $(\aGraph, (\aNode,\aNodeb))$,
let $\hat{\aGraph}$ be the graph $\aGraph$ with edges $\aLabel$ transformed into $\downarrow_{\aLabel}$, and augmented with an edge $\downarrow_{\esLabel{x}}$ such that $\aNode$ is connected with $\aNodeb$ via $\downarrow_\esLabel{x}$. Thus, it follows that:
\begin{align*}
& (\aGraph, (\aNode,\aNodeb)) \notin \supseteq\text{-}CQA(q, \Gamma) \\
\iff &  \exists \aGraph' \in \, \supseteq\text{-}Rep(\hat{\aGraph}, \Gamma) : (\aNode,\aNodeb) \not \in \semantics{\translationRPQsToGXPath{q}}_{\aGraph'}  \\
\iff & \exists \aGraph' \in \, \supseteq\text{-}Rep(\hat{\aGraph}, \aRestrictionSet') \\
\iff & \mbox{\decisionProblem{$\exists$SUPERSET-REPAIR($\hat{\aGraph}$, $\aRestrictionSet'$)} $ =$ \true}
\end{align*}
\end{proof}

Notice that this proof actually shows something stronger:

 \begin{observation}
 The problem is undecidable even when only considering restrictions from the fragment of \GRegXPath that has no node expressions and whose path expressions are of the form:

 \begin{center}
     $\aPath, \aPathb = \epsilon \mid \aLabel \mid \aPath . \aPathb 
     \mid \aPath \pathUnion \aPathb \mid \pathComplement{\aPath}$ 
 \end{center} 

 \end{observation}

As for the particular case of \decisionProblem{$\exists$SUPERSET-REPAIR} that considers restriction sets only consisting of node-expressions, we have the following result:


\begin{theorem}\label{teo:undecidabilityNodeExpressions}
There is a set of labels $\Sigma_e$ such that 
\decisionProblem{$\exists$SUPERSET-REPAIR} is undecidable, even when $\aRestrictionSet$ is a set of node expressions (but it isn't fixed). 
\end{theorem}

\begin{proof}
We reduce the problem \decisionProblem{Finite 2RPQ Entailment from $\alcoif$ KBs} to our problem. First, we give some background and the required definitions. 

Let $N_C$, $N_R$ and $N_I$ be countably infinite disjoint
sets representing \textit{concept names} (unary relations), \textit{role names} (binary relations)
and \textit{individual names} (constants), respectively. 
An \textit{assertion} is an expression of the form $C(a)$ or $r(a,b)$, where $a,b$ are individual names, $C$ is a concept name, and $r$ is a role name. 
A \textit{concept} is any of the following expressions:
\begin{equation}
    A, B = \top \mid C \mid \neg A 
     \mid A \sqcap B \mid A \sqcup B \mid \forall r. A \mid \exists r. A \mid \{a\}
\end{equation}
where $C\in N_C$, $r\in N_R$ and $a\in N_I$. 
A concept of the form $\{a\}$ is called a \textit{nominal}. 

An $\alcoif$ \textit{axiom} is a concept inclusion $A\sqsubseteq B$ or a functionality restriction ${Fun}(r)$ with $r$ a role concept. 
A \textit{knowledge base} (KB from now on) is a pair $\aKB=(\aTBox,\anABox)$ where $\aTBox$ (the \textit{TBox}) is a finite set of $\alcoif$ axioms and $\anABox$ (the \textit{ABox}) is a finite set of assertions. We denote $CN(\aKB)$ to the set of concept names in $\aKB$, $ind(\aKB)$ to the set of individuals in $\aKB$ and $nom(\aKB)$ to the set of nominals in $\aKB$. 

An \textit{interpretation} is a pair $\aKBModel=(\aKBModelDomain,\aKBModelDot)$ where $\aKBModelDomain$ is a non-empty set called the \textit{domain} of $\aKBModel$, and $\aKBModelDot$ is a mapping called the \textit{interpretation function}, that assigns as follows:
\begin{itemize}
    \item $C^\mathcal{I} \subseteq \Delta^\mathcal{I}$ for every concept name $C$;
    \item $r^\mathcal{I} \subseteq \Delta^\mathcal{I} \times \Delta^\mathcal{I}$ for every role name $r$;
    \item $a^\mathcal{I}\in \Delta^\mathcal{I}$ for every individual name $a$. 
\end{itemize}

We say that $\aKBModel$ is a \textit{finite interpretation} if $\aKBModelDomain$ is finite. For convenience, we assume $\aKBModelDomain\cap N_I=\emptyset$ for every $\aKBModel$.
The semantics for the remaining concepts can be extended as follows: 
 \begin{align*}
   \top^\aKBModel &= \aKBModelDomain \\
   (\neg A)^\aKBModel&=\aKBModelDomain\setminus A^\aKBModel\\
   (A \sqcap B)^\aKBModel&=A^\aKBModel\cap B^\aKBModel\\
   (A \sqcup B)\aKBModel &= A^\aKBModel\cup B^\aKBModel\\
   (\forall r.A)^\aKBModel&=\{ u\in \aKBModelDomain:\forall v. (u,v)\in r^\aKBModel \rightarrow v\in A^\aKBModel \}\\
   (\exists r.A)^\aKBModel&=\{ u\in \aKBModelDomain:\exists v. (u,v)\in r^\aKBModel \wedge\,  v\in A^\aKBModel \}\\
  \{a\}^\aKBModel&=\{a^\aKBModel\}
 \end{align*}

An interpretation $\aKBModel$ is a \textit{model} of $\aKB=(\aTBox,\anABox)$ if all of the following assertions hold:
\begin{itemize}
    \item  $A^\aKBModel\subseteq B^\aKBModel$ for every concept inclusion $A \sqsubseteq B$ in $\aTBox$;
    \item for every $Fun(r)$ in $\aTBox$ is true that $(u,v),(u,w)\in r^\aKBModel$ implies $v=w$, for all $u,v,w\in \aKBModelDomain$;
    \item $a^\aKBModel\in C^\aKBModel$ for every assertion $C(a)$ in $\anABox$;
    \item $(a^\aKBModel,b^\aKBModel)\in r^\aKBModel$ for every assertion $r(a,b)$ in $\anABox$.
\end{itemize}

The signature of any KB structure can be extended with new unary symbols to obtain a normal form for the concept inclusions \cite{DBLP:journals/corr/abs-1808-03130}, hence we assume without loss of generality that the TBox only contains concept inclusions of the form
$$\bigsqcap_i A_i \sqsubseteq \bigsqcup_j B_j, \quad A\sqsubseteq \forall r. B, \quad A\sqsubseteq \exists r. B, \quad A\equiv \{a\},$$
where $A_i,B_j,A,B$ are concept names and $\{a\}$ is a nominal. 

A 2RPQ is defined analogously as in (\ref{eq:rpq}), with the difference that we also allow to traverse edges backwards, i.e., we add $r^-$ to the syntax.

The problem \decisionProblem{Finite 2RPQ Entailment from $\alcoif$ KBs} asks, given a KB $\aKB$ and a 2RPQ $q$, whether $\aKBModel\models \exists x. q(x,x)$ is true for every finite model $\aKBModel$ of $\aKB$, or not. 
Notice that this problem is quite similar to superset-CQA, with the exception that the underlying structure given by the ABox is not just a graph database, but is a graph that allows multiple labels on edges and nodes, and the set of restrictions is now given by the TBox. In other words, we ask if the query is valid on any model that ``repairs'' the KB, seen as a partial representation of a graph-like model. \decisionProblem{Finite 2RPQ Entailment from $\alcoif$ KBs} is undecidable~\cite{rudolph2016undecidability}, even for a finite set $\Sigma$ of role names, which is the particular case of the problem we consider from now on.

For every KB $\aKB$ and 2RPQ $q$ we will construct a data-graph $\aDBforKB$ and a set of node constraints $\aRforKB$ such that $\aDBforKB$ has a superset-repair with respect to $\aRforKB$ if and only if there is a finite \textit{counter-model} for $\aKB$ and $\exists x.q(x,x)$, that is, a finite model of $\aKB$ that does not satisfy the query. 

We fix the set of edge labels to be $\Sigma_e=\{\down_r:r\in \Sigma\} \cup \{\down_\esLabel{total}\}$, and define the set of data values for $\aDBforKB$ as $\Sigma_n=\{P_I:I\subseteq CN(\aKB)\cup nom(\aKB)\} \cup\{T_I:I\subseteq CN(\aKB)\cup nom(\aKB)\}$. Intuitively, we will use the data values $P_I$ as a preliminary description of the concepts used in the ABox of $\aKB$, and the data values $T_I$ as a total description of the concepts used in a model of $\aKB$. 
We call \textit{partial node} to a node having data value $P_I$, and \textit{total node} to a node having data value $T_I$. In general, we assume that the indices $I$, $J$ vary in $CN(\aKB)\cup nom(\aKB)$. We say that a node \textit{contains} a concept name $A$ if $A$ lies in the set $I$ that corresponds to the index of its data value. For $\aKB=(\aTBox,\anABox)$ let us consider $\aDBforKB=(V,L,D)$ where:
\begin{itemize}
    \item $V=ind(\aKB)$;
    \item $L(a,b)=\{\downarrow_r:r(a,b)\in\anABox\}$;
    \item $D(a)=P_I$, where $I=\{C\in CN(\aKB):C(a)\in\anABox\}\cup \{ \{a\} \}$, if $\{a\} \in nom(\aKB)$, or simply $I=\{C\in CN(\aKB):C(a)\in\anABox\}$ if otherwise. 
\end{itemize}

We now present the formulas in $\aRforKB$, together with a short description of the restriction we want to model in each case:
\begin{enumerate}[label=(\roman*)]
    \item\label{eq:assign} The conjunction of the following formulas
    
    \begin{align*}
    	&\bigwedge_I (P_I^{\,=} \Rightarrow  \langle \down_\esLabel{total} \big[\bigvee_{I\subseteq J}  T_J^{\,=}\big] \rangle)\\
    	&\neg\langle \down_\esLabel{total}^-\down_\esLabel{total} \cap\, \overline{\epsilon} \rangle \\
    	&\langle\down_\esLabel{total}\rangle \Rightarrow  \bigvee_{J}  P_J^{\,=}
    \end{align*}
    We denote this conjunction by $\psi_1$. This formula states that every partial node must be connected to precisely one total node through the relation $\down_\esLabel{total}$, and every total edge in the data-graph must be an outgoing edge from a partial node. Furthermore, the middle part of the formula asserts that partial nodes cannot have more than one outgoing total edge. A total node connected to a partial node through the relation $\down_\esLabel{total}$ is called \textit{the total counterpart} of that partial node. 
    \item\label{eq:ci1} For every $\bigsqcap_i A_i \sqsubseteq \bigsqcup_j B_j$ in $\aTBox$ consider: 
    $$\psi_2=\bigvee_{\forall i. \,A_i\in I} T_I^{\,=} \Rightarrow \bigvee_{\exists j.\,B_j\in I} T_J^{\,=},$$
    which states that the data value of every total node must contain all the concept-name information to describe a model of $\aKB$.
    \item\label{eq:ci2} 
    For every $A\sqsubseteq \forall r. B$ in $\aTBox$ consider: 
    $$\psi_3=\neg \langle \big[\bigvee_{A\in I} T_I^{\,=} \big] (\epsilon\,\cup\down_\esLabel{total}^-) \downarrow_r (\epsilon\,\cup\down_\esLabel{total}) \big[\bigvee_{B\not\in J} T_J^{\,=}\big] \rangle,$$
    which states that, for every total node containing the concept name $A$ such that either it has an outgoing edge $\down_r$ or the node itself is the total counterpart of a partial node having an outgoing edge $\down_r$, is not possible that the total node meets another total node that does not contain the concept $B$ through a path in $(\epsilon\,\cup\down_\esLabel{total}^-) \down_r (\epsilon\,\cup\down_\esLabel{total})$. One way to intuitively interpret this, is to think of every outgoing edge from a partial node also as an outgoing edge from its total counterpart, and vice versa. From now on, ``outgoing edge'' will have this interpretation. 
    \item\label{eq:ci3} For every $A\sqsubseteq \exists r. B$ in $\aTBox$ consider: 
    $$\psi_4=\bigvee_{A\in I} T_I^{\,=} \Rightarrow \langle (\epsilon\,\cup\down_\esLabel{total}^-)\down_r(\epsilon\,\cup\down_\esLabel{total})\big[\bigvee_{B\in J}  T_J^{\,=}\big] \rangle,$$
    which states that every total node containing the concept name $A$ must have an outgoing edge $\down_r$ joining it with another total node containing concept $B$. 
    \item\label{eq:ci4} For every $A\equiv\{a\}$ in $\aTBox$, we consider the conjunction of the following formulas: 
    \begin{align*}
    	&\psi_5=\bigvee_{\{a\}\in I} P_I^{\,=}\Rightarrow \langle\down_\esLabel{total}\big[\bigvee_{A\in J} T_J^{\,=}\big] \rangle \\
    	&\psi_6=\bigvee_{A\in J} T_J^{\,=}\Rightarrow \langle\down_\esLabel{total}^-\big[\bigvee_{\{a\}\in I}  P_I^{\,=}\big] \rangle
    \end{align*} 
    which states that a total node contains concept $A$ if and only if such node is the total counterpart of a partial node containing the nominal $\{a\}$.
    \item\label{eq:fun} For every $Fun(r)$ in $\aTBox$, we consider the conjunction of the following formulas: 
    \begin{align*}
    	& \psi_7=\neg\langle\big[\bigvee_I T_I^{\,=}\big]\down_r^-\down_r\big[\bigvee_I T_I^{\,=}\big]\cap\, \overline{\epsilon}\rangle \\
    	& \psi_8=\neg\langle\big[\bigvee_I T_I^{\,=}\big]\down_r^-\down_\esLabel{total}^-\down_r(\epsilon\,\cup\down_\esLabel{total})\big[\bigvee_I T_I^{\,=}\big]\cap\, \overline{\epsilon}\rangle \\
    	& \psi_9=\neg\langle\big[\bigvee_I T_I^{\,=}\big]\down_\esLabel{total}^-\down_r^-\down_r(\epsilon\,\cup\down_\esLabel{total})\big[\bigvee_I T_I^{\,=}\big]\cap\,
    	\overline{\epsilon}\rangle
    \end{align*}
    
    Notice that these three formulas could be unified into a single one. However, for simplicity we write them separately to help us clarify which patterns (and how) we wish to avoid in a valid repair, which are depicted in Figure~\ref{fig:patterns_to_avoid}. 

\begin{figure}[H]
\centering
\begin{tabular}{ l | l | l | l | l }
	(a) & (b) & (c) & (d) & (e) \\
	\begin{tikzpicture}%
		[>=stealth,
		shorten >=1pt,
		node distance=1cm,
		on grid,
		auto,
		every state/.style={draw=black!60, fill=black!5, very thick}
		]
		\node [circle,draw,fill] (mid)                  { };
		\node  [rectangle,fill,draw] (upper) [above right=of mid] { };
		\node [rectangle,draw] (right) [right=of mid] { };
		\node [circle] (upper2) [above right=of upper] { };
		
		\path[-latex]
		(mid)   edge  node[swap]                      {$r$} (right)
		(mid)   edge  node                      {$r$} (upper)
		;
	\end{tikzpicture}
	&
	\begin{tikzpicture}%
		[>=stealth,
		shorten >=1pt,
		node distance=1cm,
		on grid,
		auto,
		every state/.style={draw=black!60, fill=black!5, very thick}
		]
		\node [circle,draw] (mid)                  { };
		\node  [rectangle,fill,draw] (upper) [above right=of mid] { };
		\node [rectangle,draw] (right) [right=of mid] { };
		\node [rectangle,draw,fill] (upper2) [above right=of upper] { };
		\node [circle] (right2) [above right=of right] { };
			
		\path[-latex]
		(mid)   edge  node[swap]                      {$r$} (right)
		(mid)   edge  node                      { } (upper)
		(upper)   edge  node                      {$r$} (upper2)
		;
	\end{tikzpicture}
	&
	\begin{tikzpicture}%
		[>=stealth,
		shorten >=1pt,
		node distance=1cm,
		on grid,
		auto,
		every state/.style={draw=black!60, fill=black!5, very thick}
		]
		\node [circle,draw] (mid)                  { };
		\node  [rectangle,draw,fill] (upper) [above right=of mid] { };
		\node [circle,draw] (right) [right=of mid] { };
		\node [rectangle,draw,fill] (upper2) [above right=of upper] { };
		\node [rectangle,draw] (right2) [above right=of right] { };
			
		\path[-latex]
		(mid)   edge  node[swap]                      {$r$} (right)
		(mid)   edge  node                    { } (upper)
		(upper)   edge  node                      {$r$} (upper2)
		(right)   edge  node                    { } (right2)
		;
	\end{tikzpicture}
	&
	\begin{tikzpicture}%
		[>=stealth,
		shorten >=1pt,
		node distance=1cm,
		on grid,
		auto,
		every state/.style={draw=black!60, fill=black!5, very thick}
		]
		\node [circle,draw,fill] (mid)                  { };
		\node  [circle,draw] (upper) [above right=of mid] { };
		\node [rectangle,draw] (right) [right=of mid] { };
		\node [rectangle,draw,fill] (upper2) [above right=of upper] { };
		\node [circle] (right2) [above right=of right] { };
		
		\path[-latex]
		(mid)   edge  node[swap]                      {$r$} 
		(right)
		(mid)   edge  node                    {$r$} (upper)
		(upper)   edge  node                      {} (upper2)
		;
	\end{tikzpicture}
	&
	\begin{tikzpicture}%
		[>=stealth,
		shorten >=1pt,
		node distance=1cm,
		on grid,
		auto,
		every state/.style={draw=black!60, fill=black!5, very thick}
		]
		\node [circle,draw,fill] (mid)                  { };
		\node  [circle,draw] (upper) [above right=of mid] { };
		\node [circle,draw] (right) [right=of mid] { };
		\node [rectangle,draw,fill] (upper2) [above right=of upper] { };
		\node [rectangle,draw] (right2) [above right=of right] { };
		
		\path[-latex]
		(mid)   edge  node[swap]                      {$r$} (right)
		(mid)   edge  node                    {$r$} (upper)
		(upper)   edge  node                      { } (upper2)
		(right)   edge  node                    { } (right2)
		;
	\end{tikzpicture}
	\\ \hline
	&&&& \\
	\begin{tikzpicture}%
		[>=stealth,
		shorten >=1pt,
		node distance=1cm,
		on grid,
		auto,
		every state/.style={draw=black!60, fill=black!5, very thick}
		]
		\node [circle,draw,fill] (mid)                  { };
		\node  [rectangle,fill,draw] (upper) [above right=of mid] { };
		\node [rectangle] (right) [right=of mid] { };
		\node [circle] (upper2) [above right=of upper] { };
		
		\path[-latex]
		(mid)   edge  node                      {$r$} (upper)
		;
	\end{tikzpicture}
	&
	\begin{tikzpicture}%
		[>=stealth,
		shorten >=1pt,
		node distance=1cm,
		on grid,
		auto,
		every state/.style={draw=black!60, fill=black!5, very thick}
		]
		\node [circle,draw] (mid)                  { };
		\node  [rectangle,draw,fill] (upper) [above right=of mid] { };
		\node [rectangle] (right) [right=of mid] { };
		\node [rectangle,draw,fill] (upper2) [above right=of upper] { };
		\node [circle] (right2) [above right=of right] { };
		
		\path[-latex]
		(mid)   edge[bend right=50]  node[swap]                      {$r$} (upper2)
		(mid)   edge  node                    { } (upper)
		(upper)   edge  node                      {$r$} (upper2)
		;
	\end{tikzpicture}
	&
	\begin{tikzpicture}%
		[>=stealth,
		shorten >=1pt,
		node distance=1cm,
		on grid,
		auto,
		every state/.style={draw=black!60, fill=black!5, very thick}
		]
		\node [circle,draw] (mid)                  { };
		\node  [rectangle,draw,fill] (upper) [above right=of mid] { };
		\node [circle,draw] (right) [right=of mid] { };
		\node [rectangle,draw,fill] (upper2) [above right=of upper] { };
		\node [circle] (right2) [above right=of right] { };
		
		\path[-latex]
		(mid)   edge  node[swap]                      {$r$} (right)
		(mid)   edge  node                    { } (upper)
		(upper)   edge  node                      {$r$} (upper2)
		(right)   edge[bend right=10]  node                      { } (upper2)
		;
	\end{tikzpicture}
	&
	\begin{tikzpicture}%
		[>=stealth,
		shorten >=1pt,
		node distance=1cm,
		on grid,
		auto,
		every state/.style={draw=black!60, fill=black!5, very thick}
		]
		\node [circle,draw,fill] (mid)                  { };
		\node  [circle,draw] (upper) [above right=of mid] { };
		\node [rectangle] (right) [right=of mid] { };
		\node [rectangle,draw,fill] (upper2) [above right=of upper] { };
		\node [circle] (right2) [above right=of right] { };
		
		\path[-latex]
		(mid)   edge  node                      {$r$} (upper)
		(upper)   edge  node                      { } (upper2)
		(mid)   edge[bend right=50]  node[swap]                      {$r$} (upper2)
		;
	\end{tikzpicture}
	&
	\begin{tikzpicture}%
		[>=stealth,
		shorten >=1pt,
		node distance=1cm,
		on grid,
		auto,
		every state/.style={draw=black!60, fill=black!5, very thick}
		]
		\node [circle,draw,fill] (mid)                  { };
		\node  [circle,draw] (upper) [above right=of mid] { };
		\node [circle,draw] (right) [right=of mid] { };
		\node [rectangle,draw,fill] (upper2) [above right=of upper] { };
		\node [circle] (right2) [above right=of right] { };
		
		\path[-latex]
		(mid)   edge  node                      {$r$} (upper)
		(mid)   edge  node[swap]                      {$r$} (right)
		(upper)   edge  node                      { } (upper2)
		(right)   edge[bend right=20]  node                      { } (upper2)
		;
	\end{tikzpicture}
	
\end{tabular}
\caption{In the upper boxes, the patterns we do not wish to appear in a repair as specified by the constraints. In the lower boxes, how to fix each problem by merging the problematic $\square$ node with a $\blacksquare$ node. The nodes are represented as follows: $\square$ undesired total nodes, $\blacksquare$ total nodes, $\fullmoon$ partial nodes, $\newmoon$ any kind of nodes. The unlabeled edges represent total edges, thus the $\square$ or $\blacksquare$ node is the total counterpart of the incident $\fullmoon$ node. The case (a) is handled by $\psi_7$, cases (b) and (c) by $\psi_8$, and cases (d) and (e) by $\psi_9$.}
\label{fig:patterns_to_avoid}
\end{figure}
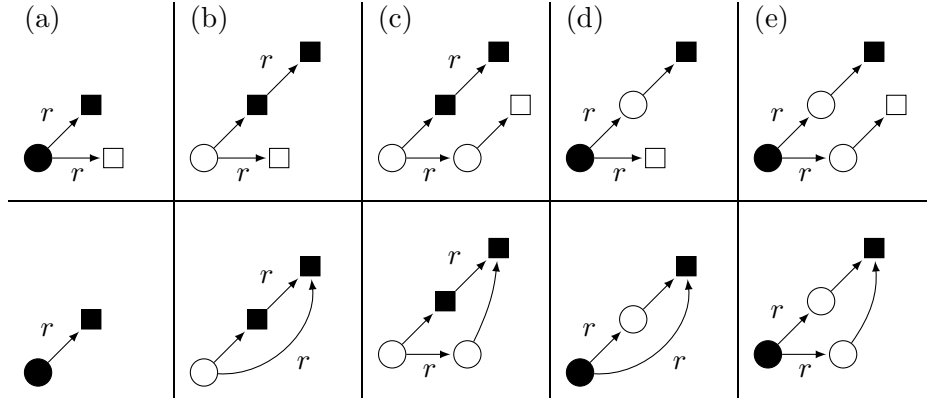

    \item\label{eq:notquery} Finally, we should add a restriction that ensures the non-satisfiability of the query. But first, we show how to translate a 2RPQ $\alpha$ to obtain the path expression $\translationRPQsToGXNode{\alpha}$ that we need. The translation is as follows:
\begin{itemize}[leftmargin=.6in,label={}]
\itemsep0em 
  \item $\epsilon\mapsto \epsilon\,\cup\down_\esLabel{total}\cup\down_\esLabel{total}^-$, denoted by $\translationRPQsToGXNode{\epsilon}$
  \item $r \mapsto \translationRPQsToGXNode{\epsilon}\down_r \translationRPQsToGXNode{\epsilon}$
  \item $r^- \mapsto \translationRPQsToGXNode{\epsilon}\down_r^-\translationRPQsToGXNode{\epsilon}$
  \item $\alpha^* \mapsto (\translationRPQsToGXNode{\alpha})^*$
  \item $\alpha \cdot \beta \mapsto \translationRPQsToGXNode{\alpha} \cdot \translationRPQsToGXNode{\beta}$
  \item $\alpha \cup \beta \mapsto \translationRPQsToGXNode{\alpha} \cup \translationRPQsToGXNode{\beta}$
\end{itemize}

We add the formula $\psi_{10}=\langle \epsilon \cap \overline{\translationRPQsToGXNode{q}}\rangle$ to $\aRforKB$. Notice that if we delete every appearance of the symbol $\down_\esLabel{total}$ in a word $\omega$ from $\translationRPQsToGXNode{\alpha}$ we obtain a word from $\translationRPQsToGXPath{\alpha}$, where $\translationRPQsToGXPath{\alpha}$ is the translation of $\alpha$ as defined in the proof of Theorem \ref{teo:supersetIndecidible}.
\end{enumerate}

We now proceed to prove that there is a finite counter-model for $\aKB$ and $\exists x.q(x,x)$ if and only if there is a superset-repair for $\aDBforKB$ and $\aRforKB$. 

$\implies$) Let $\aKBModel=(\aKBModelDomain,\aKBModelDot)$ be a counter-model for $\aKB$ and $\exists x.q(x,x)$. We define the data-graph $G'=(V',L',D')$, where $V'=ind(\aKB)\cup \aKBModelDomain=V\cup\aKBModelDomain$, $D'|_V=D$, $L'|_{V\times V}=L$, for $u\in \aKBModelDomain$, $D'(u)=T_I$ with $I=\{A\in CN(\aKB):u\in A^\aKBModel\} \cup \{\{a\}\in nom(\aKB):a^\aKBModel=u\}$, and 
\begin{itemize}[leftmargin=.6in,label={}]
\itemsep0em 
  \item $L'(a,a^\aKBModel)=\{\down_\esLabel{total}\}$, for every $a\in V$,
  \item $L'(u,v)=\{\down_r:(u,v)\in r^\mathcal{I}\}$, for every $u,v\in \aKBModelDomain$,
  \item and $L'(u,v)$ is empty for any other remaining case.
\end{itemize}

Notice that if $r(a,b)$ is an assertion, then $\down_r \in L'(a,b)\cap L'(a^\aKBModel,b^\aKBModel)$, and if $\{a\}$ is a nominal in $\aKB$, then $\semantics{\bigvee_{\{a\}\in I} P_I^{\,=}}_{G'}=\{a\}$. This is relevant to simplify the proof below. 

We show that the graph $G'$ satisfies each of the formulas in $\aRforKB$.
It follows by definition that formula in item~\ref{eq:assign} is satisfied. 

For item \ref{eq:ci1}, suppose that $\bigsqcap_i A_i \sqsubseteq \bigsqcup_j B_j$ is in $\aTBox$ and $u\in V'$ belongs to $\semantics{T_I^{\,=}}_{G'}$ where $A_i\in I$ for every $i$. This implies by definition that $u\in\aKBModelDomain$ and $u\in A_i^\aKBModel$ for every $i$. Since $\aKBModel\models\aTBox$, we have that $u\in B_k^\aKBModel$ for some $k$, hence $B_k\in I$. 

For item \ref{eq:ci2}, suppose that $A\sqsubseteq \forall r. B$ is in $\aTBox$ and let $u\in V'$. If $u$ does not contain the concept $A$, then $u$ is in $\semantics{\psi_3}_{G'}$. On the other hand, if $D'(u)=T_I$ for $A\in I$, and $(u,v)\in r^\aKBModel$ for some $v\in \aKBModelDomain$, then we have that $v\in B^\aKBModel$ since $\aKBModel\models\aTBox$, and thus $D'(v)=T_J$ for some $B\in J$. It follows that $u$ is in $\semantics{\psi_3}_{G'}$. 

For item \ref{eq:ci3}, suppose $A\sqsubseteq \exists r. B$ is in $\aTBox$ and $u\in V'$ is a total node that contains the concept $A$. This implies by definition that $u\in\aKBModelDomain$ and $u\in A^\aKBModel$. Since $\aKBModel\models\aTBox$, there is $v\in \aKBModelDomain$ such that $(u,v)\in r^\aKBModel$ and $v\in B^\aKBModel$. Thus, $u\in \semantics{\psi_4}_{G'}$. 

For item \ref{eq:ci4}, suppose $A\equiv\{a\}$ is in $\aTBox$ and let $u=a^\aKBModel$. Since $\aKBModel\models\aTBox$, it follows that $u\in A^\aKBModel$. This implies that $a\in\semantics{\psi_5}_{G'}$. 
Now, if $v\in V'$ is a total node containing the concept $A$, then $v\in A^\mathcal{I}$, thus by hypothesis, $v=u$. It follows that 
$\semantics{\bigvee_{A\in J } T_J^{\,=}}_{G'}=\{u\}$ and hence $\semantics{\psi_6}_{G'}=V'$. 

For item \ref{eq:fun}, suppose $Fun(r)$ is in $\aTBox$. We will only show that $\semantics{\psi_9}_{G'}=V'$, since a similar argument holds for the remaining formulas of the conjunction. 
If $u\in ind(\aKB)$ or $u\neq b^\aKBModel$ for every $b\in ind(\aKB)$, then it is straightforward that $u\in\semantics{\psi_9}_{G'}$. If instead $u=b^\aKBModel$ for $b\in ind(\aKB)$, and $r(a,b),r(a,c)$ are assertions in $\aKB$, then we assert that $b^\aKBModel=c^\aKBModel$ since $\aKBModel$ is a model of $\aKB$. This implies $u\in\semantics{\psi_9}_{G'}$. 
 
For item \ref{eq:notquery} we will prove that, if $\semantics{\psi_{10}}_{G'}\neq V'$, then $\aKBModel$ satisfies the query, i.e. there is a cycle in $\aKBModel$ reading a word from the language $q$. 
First, notice that for every 2RPQ $\alpha$, $a\in V$ and $u\in V'$, the pair $(a,u)\in\semantics{\translationRPQsToGXNode{\alpha}}$ if and only if $(a^\aKBModel,u)\in\semantics{\translationRPQsToGXNode{\alpha}}$, and the pair $(u,a)\in\semantics{\translationRPQsToGXNode{\alpha}}$ if and only if $(u,a^\aKBModel)\in\semantics{\translationRPQsToGXNode{\alpha}}$. Since $\down_r \in L'(a,b)\cap L'(a^\aKBModel,b^\aKBModel)$ for every assertion $r(a,b)$ in $\aKB$, we conclude that, if $u,v\in\aKBModelDomain$ and $(u,v)\in \semantics{\translationRPQsToGXNode{\alpha}}$, then $(u,v)\in \semantics{\translationRPQsToGXPath{\alpha}}$. In other words, if $u$ and $v$ are connected by a path with label in $\translationRPQsToGXNode{\alpha}$, then $u$ and $v$ are also connected by a path that has no edge $\down_\esLabel{total}$. 
This statement can be easily proved by induction on the structure of $\alpha$. Suppose now that $\semantics{\psi_{10}}_{G'}\neq V'$, and let $u\in V'$ be a node for which the constraint $\psi_{10}$ is not valid. By definition, this occurs when $(u,u)\in \semantics{\translationRPQsToGXNode{q}}_{G'}$. We may assume that $u\in\aKBModelDomain$, thus $(u,u)\in\semantics{\translationRPQsToGXPath{q}}_{G'}$, which implies that $\aKBModel$ has a cycle reading a word from $q$ that contains the node $u$. 

Therefore, $G'$ is a finite data-graph that contains $\aDBforKB$ and satisfies $\aRforKB$, hence it contains a superset-repair of $\aDBforKB$. 

$\impliedby$) Let $G'=(V',L',D')$ be a superset-repair of $\aDBforKB$ with respect to $\aRforKB$. We will obtain a counter-model for $\aKB$ and $\exists x.q(x,x)$ by applying an equivalence relation on $V'$. Let $\sim$ be the relation on $V'$ defined as: $u\sim v$ if and only if $(u,v)\in \semantics{(\down_\esLabel{total}\cup\down_\esLabel{total}^-)^*}_{G'}$. We use the Kleene star to guarantee that the relation is transitive and reflexive, and the symmetry is a direct consequence of the definition. In essence, we are collapsing every node having data value $P_I$ with its total counterpart, which we know exists since \ref{eq:assign} is satisfied. For every $u\in V'$, we denote by $[u]$ to its equivalence class. 
Let us consider $\aKBModel=(\aKBModelDomain,\aKBModelDot)$, where $\aKBModel=V'/\sim$ and $\aKBModelDot$ is defined over the signature of $\aKB$ as follows:
\begin{itemize}
    \item $C^\aKBModel=\{[u]:\exists I,\exists u'\sim u. \, D'(u')=T_I \text{ and } C\in I\}$ for every concept name $C$; 
    \item $r^\aKBModel = \{ ([u],[v]):\exists u'\sim u, v'\sim v.\, \down_r\in L'(u',v')\}$ for every role name $r$;
    \item $a^\aKBModel = [a]$ for every individual name $a$. 
\end{itemize} 

Notice that every class $[u]$ contains exactly one total node, and it may contain several (or no) elements from $ind(\aKB)$ and other partial nodes. We denote by $\totalrep{u}$ to the total node of the class $[u]$.

We will now show that $\aKBModel$ is a counter-model for $\aKB=(\aTBox,\anABox)$ and $\exists x.q(x,x)$. It follows from the definition that $a^\aKBModel\in C^\aKBModel$ for every assertion $C(a)$ in $\anABox$, and $(a^\aKBModel,b^\aKBModel)\in r^\aKBModel$ for every assertion $r(a,b)$ in $\anABox$. 

Suppose that $\bigsqcap_i A_i \sqsubseteq \bigsqcup_j B_j$ is in $\aTBox$, and let $[u]\in \aKBModelDomain$ such that $[u]\in A_i^\aKBModel$ for every $i$. 
For the element $u'=\totalrep{u}$, we know that $D'(u')=T_I$ with $A_i\in I$ for every $i$. Since $G'$ satisfies the constraint $\psi_2$ from item \ref{eq:ci1}, it follows that $B_k\in I$ for some $k$, which implies that $[u]\in B_k^\aKBModel$. 

Suppose that $A\sqsubseteq \forall r. B$ is in $\aTBox$, and let $[u],[v]\in \aKBModelDomain$ such that $[u]\in A^\aKBModel$ and $([u],[v])\in r^\aKBModel$. Thus, $D'(\totalrep{u})=T_I$ where $A\in I$, and there are nodes $u'$ and $v'$, such that $u'\sim u,v'\sim v$, and $\down_r\in L'(u',v')$. 
Since $G'$ satisfies $\psi_3$ from item \ref{eq:ci2}, if $D'(v')=T_J$ for some $J$ then $B$ must be in $J$. If instead $D'(v')=P_J$, then $(\totalrep{u},\totalrep{v})\in \semantics{(\epsilon\,\cup\down_\esLabel{total}^-)\down_r \down_\esLabel{total}}_{G'}$ and thus $D'(\totalrep{v})=T_{J'}$ for $B\in J'$. Whichever the case may be, we obtain $[v]\in B^\mathcal{I}$. 

Suppose that $A\sqsubseteq \exists r. B$ is in $\aTBox$, and let $[u]\in \aKBModelDomain$ such that $[u]\in A^\aKBModel$. Thus $D'(\totalrep{u})=T_I$ where $A\in I$, and since $G'$ satisfies $\psi_4$ from item \ref{eq:ci3}, there is a node $v$ such that $D'(v)=T_J$ with $B\in J$ and $(\totalrep{u},v)\in \semantics{(\epsilon\,\cup\down_\esLabel{total}^-)\down_r(\epsilon\,\cup\down_\esLabel{total})}_{G'}$. This implies that $([u],[v])\in r^\aKBModel$ and $[v]\in B^\aKBModel$.  

Suppose that $A\equiv\{a\}$ is in $\aTBox$, and let $[u]\in \aKBModelDomain$ such that $[u]\in A^\aKBModel$. Hence, $D'(\totalrep{u})=T_I$ where $A\in I$, and since $G'$ satisfies $\psi_5$ and $\psi_6$ from item \ref{eq:ci4}, it follows that $\totalrep{u}=a^\aKBModel$. 

Suppose that $Fun(r)$ is in $\aTBox$ and $([u],[v]),([u],[w])\in r^\aKBModel$. We need to prove that $[v]=[w]$ or, equivalently, that $\totalrep{v}=\totalrep{w}$. By hypothesis, there exist $u'\sim u,u''\sim u, v'\sim \totalrep{v}$ and $w'\sim \totalrep{w}$ such that $\down_r\in L'(u',v')$ and $\down_r\in L'(u'',w')$. There are several cases to analyse. If $u'=u''$, $v'=\totalrep{v}$ and $w'=\totalrep{w}$, it follows that $v'=w'$ from the satisfaction of formula $\psi_7$ from item \ref{eq:fun}. For the remaining cases, consider that every node that is different from its total counterpart is connected to it by a total edge. As a consequence of this fact and the satisfaction of formulas $\psi_8$ and $\psi_9$, we obtain that the statement holds for the remaining cases. 

Finally, we have to check that $\aKBModel$ does not satisfy the query. Toward a contradiction, let $[u]\in \aKBModelDomain$ such that, for a cycle in $\aKBModel$ containing $[u]$, the label of the path starting and finishing in $[u]$ reads a word on the language $q$. 
We will prove that there is a cycle in $G'$ reading a word from the language $\translationRPQsToGXNode{q}$. Moreover, this cycle contains $\totalrep{u}$ and the read word starts in this node. Suppose the cycle in $\aKBModel$ is 
$$([u_0],r_1,[u_1]),([u_1],r_2,[u_2]),\ldots, ([u_{n-1}],r_n,[u_n]),$$
where $[u_0]=[u_n]=[u]$ and $\omega=r_1r_2\cdots r_n\in q$. The triplet $(v,\alpha,w)$ means $(v,w)\in \alpha^\aKBModel$. Then, $$(\totalrep{u_0},\translationRPQsToGXNode{r_1},\totalrep{u_1}), (\totalrep{u_1},\translationRPQsToGXNode{r_2},\totalrep{u_2}), \ldots, (\totalrep{u_{n-1}},\translationRPQsToGXNode{r_n},\totalrep{u_n})$$
is a cycle in $G'$. Since any word in $\translationRPQsToGXNode{r_1}\translationRPQsToGXNode{r_2}\cdots \translationRPQsToGXNode{r_n}$ is also a word in $\translationRPQsToGXNode{q}$ we obtain that $G'$ cannot satisfy $\psi_{10}$ since $(\totalrep{u},\totalrep{u})\in \semantics{\translationRPQsToGXNode{q}}_{G'}$.

\end{proof}

As pointed out in Remark~\ref{remark:20}, simple data tests may prevent the existence of superset repairs. Observe that in the absence of them, any data-graph $\aGraph$ has a superset repair if we consider only the positive fragment of \Gregxpath: add every possible edge label $l \in \Sigma_e$ to every pair of nodes $v, v' \in V_\aGraph$ and the resulting graph will satisfy any expression. The proof of this fact follows quite straightforward by induction in the expression's structure.

Before proceeding, we define some concepts related to data values:

\begin{definition}
Let $\eta$ be a \Gregxpath expression. We define the set of data values present in $\eta$ as the set of all those $c \in \Sigma_n$ such that the subexpression $[c^=]$ or $[c^{\neq}]$ is used in $\eta$. We denote it as \defstyle{$\Sigma_n^\eta$}.

Analogously, we define the set of data values used by a set of $\Gregxpath$ expressions $\aRestrictionSet$ as $\defstyle{\Sigma_n^\aRestrictionSet} = \bigcup\limits_{\eta \in \aRestrictionSet} \Sigma_n^\eta$.

We also denote the set of data values used in a graph $\aGraph$ as \defstyle{$\Sigma_n^\aGraph$}.
\end{definition}

Even though $\Sigma_n$ may be infinite, when considering only \Gposregxpath expressions we obtain the following lemma:

\begin{lemma}\label{teo:dataValuesInR}
Let $\aGraph$ be a data-graph and $\aRestrictionSet$ a set of \Gposregxpath expressions. If there is a superset repair $\aGraph'$ of $\aGraph$ with respect to the constraints $\aRestrictionSet$, then there is another superset repair $H$ that only uses data values from $\Sigma_n^\aRestrictionSet \cup \Sigma_n^\aGraph$ plus at most two extra data values not mentioned in $\aRestrictionSet$.
\end{lemma}

\begin{proof} 
Intuitively, values that are not mentioned in $\aRestrictionSet$ are not neccessary to satisfy $\aRestrictionSet$, and therefore, if there exists a repair, there must be one with only data values from $\aRestrictionSet$ and the original data values from $\aGraph$. Nonetheless, it could be the case that every data value mentioned in $\aRestrictionSet$ is in an expression of the form $[c^{\neq}]$, and in that case we might need an extra data value (for example, consider the expression $\aFormula = \comparacionCaminos{\gamma}$ where $\gamma=\down_x [\bigvee\limits_{c \in \mathcal{D}}c^{\neq}]$ for $\mathcal{D}$ a set of data values). Actually, we might need two fresh data values to satisfy some expression of the form $\comparacionCaminos{\aPath \neq \aPathb}$ (replace $\aPath = \aPathb = \gamma$). See the Appendix~\ref{Appendix} for the detailed proof.
\end{proof}

The following observation is a consequence of the previous lemma:

\begin{observation}
If there is a superset repair $\aGraph$ with respect to $\aRestrictionSet$, then there exists a superset repair with a number of different data values that linearly depends on $|\aGraph| + |\aRestrictionSet|$.
\end{observation}

Notice that this observation does not imply that there must be a superset repair with linear size on $|\aGraph| + |\aRestrictionSet|$: it could be the case that there is an exponential number of nodes having the same data value. However, we could somehow `merge' all those nodes with the same data value while preserving some edges, such that the resulting graph will still satisfy all those $\Gposregxpath$ expressions that were satisfied in the original graph:

\begin{lemma} \label{teoremaHorrible}
Let $\aGraph$ be a data-graph that satisfies a \Gposregxpath restriction $\eta$. Let $V_{d}$ be a set of nodes from $\aGraph$ having the same data value $d$. If we define a new data-graph $H$ where
\begin{align*}
V_H &= (V_\aGraph \setminus V_d) \cup \{v_d\} \\ 
L_H(v,w) &= L_\aGraph(v,w) \forall v,w \in V_\aGraph \setminus V_d \\
L_H(v,v_d) &= \{e \in \Sigma_n \mid \exists w \in V_d \mbox{ such that } e \in L_\aGraph(v,w)\}, \forall v \in V_\aGraph \setminus V_d\\
L_H(v_d,v) & = \{e \in \Sigma_n \mid \exists w \in V_d \mbox{ such that } e \in L_\aGraph(w,v)\}, \forall v \in V_\aGraph \setminus V_d \\
L_H(v_d,v_d) & = \{e \in \Sigma_n \mid \exists w_1,w_2 \in V_d \mbox{ such that } e \in L(w_1,w_2)\}\\
D_H(v) &= D_\aGraph(v), \forall v \in V_\aGraph \setminus V_d \\
D_H(v_d) &= d
\end{align*}

Then, $H$ satisfies $\eta$.

\end{lemma}

\begin{proof}
Expressions from $\Gposregxpath$ can only interact with data values and edge labels, ignoring the actual identity of the nodes\footnote{This is not completely true, since the path expression $\epsilon$ can be used to relate a node with only itself.}. Then, observe that in the data-graph just defined the neighbourhoods of all nodes are the same as in data-graph $G$, except that we might have collapsed some set of nodes. But for those nodes collapsed, we created a new one with the same data value and the same (or even bigger) neighbourhood. Finally, the positive expression of \Gposregxpath cannot really distinguish this new node from the previous ones. See the Appendix~\ref{Appendix} for the actual proof.
\end{proof}

As it was mentioned previously, when we define the graph $H$ the set of nodes $V_d$ is collapsed into a unique node $v_d$ while preserving all the edges that were incident to the set $V_d$. This operation is usually called vertex contraction. Using Lemmas~\ref{teo:dataValuesInR} and \ref{teoremaHorrible}, we prove the following very useful fact:

\begin{theorem} \label{teo:superset}
Let $\aGraph$ be a data-graph and $\aRestrictionSet$ a set of \Gposregxpath expressions. If there is a superset repair of $\aGraph$ with respect to $\aRestrictionSet$, then there is a superset repair of $\aGraph$ with respect to $\aRestrictionSet$ of polynomial size depending on $|\aGraph| + |\aRestrictionSet|$.
\end{theorem}

\begin{proof}
We define a data-graph $H$ such that $\aGraph \subseteq H$, $H \models \aRestrictionSet$, and the size of $H$ polynomially depends on $|\aGraph| + |\aRestrictionSet|$. This will suffice to prove the theorem.

Let $\aGraph'$ be a superset repair of $\aGraph$ with respect to $\aRestrictionSet$ that uses a number of data values that linearly depends on $|\aGraph| + |\aRestrictionSet|$. Such a superset repair exists as a consequence of Lemma~\ref{teo:dataValuesInR} and the hypothesis. Thus, for every data value $c \in \Sigma_n^{\aGraph'} \setminus \Sigma_n^\aGraph$, we contract all the nodes having data value $c$ to a unique node. Let $H'$ be the graph obtained once those nodes were contracted. Notice that $H'$ satisfies $\aRestrictionSet$ due to Lemma~\ref{teoremaHorrible}.

We contract every other node $v \in V_{\aGraph'} \setminus V_{\aGraph}$ in $H'$ that has a data value from $\Sigma_n^\aGraph$ with a node from $\aGraph$ having the same data value. The resulting graph is indeed the $H$ we are looking for, since once again due to Lemma~\ref{teoremaHorrible}, $H$ satisfies $\aRestrictionSet$ and $|H|$ is exactly $|V_\aGraph| + |\Sigma_n^{\aGraph'} \setminus \Sigma_n^\aGraph|$. Moreover, $|H|$ linearly depends on $|\aGraph| + |\aRestrictionSet|$. Therefore, the size of $H$ is at most quadratic on $|\aGraph| + |\aRestrictionSet|$.

\end{proof}

This last theorem shows that the problem $\exists$\textsc{SUPERSET-REPAIR} lies in \textsc{NP} for \Gposregxpath expressions, since we can use the repair itself as a witness for a positive instance. Moreover, we can use the same proof to define a polynomial-time algorithm that computes a superset repair of a data-graph $\aGraph$. If there exists a superset repair of $\aGraph$ with respect to $\aRestrictionSet$, then there is a `small' graph $H$ such that $H \models R$ and $H$ has a very precise structure: the only nodes that are added to $\aGraph$ in order to obtain $H$ have a one-to-one correspondence with those new data values that were not present in $\aGraph$. From now on we will refer to this structure as the \textit{standard form of a superset repair}. Hence, we may find a minimal --with respect to node inclusion-- data-graph $H'$ that satisfies $H' \models \aRestrictionSet$ and $G \subseteq H'$ as follows: 
Iterate over every possible subset $S$ of $(\Sigma_n^\aRestrictionSet \cup \{c,d\}) \setminus \Sigma_n^{\aGraph}$, where $c$ and $d$ are the `fresh' data values mentioned in Lemma~\ref{teo:dataValuesInR}. Then, add one node to $\aGraph$ for each data value in $S$ and every possible edge between any pair of nodes. Finally, check if the resulting data-graph satisfies $\aRestrictionSet$.

If none of these data-graphs satisfies $R$, then it follows from Theorem~\ref{teo:superset} that there is no superset repair. Otherwise, after computing the minimal --with respect to node inclusion-- data-graph $H'$, we may find a repair by deleting edges from $H'$: if once we delete an edge $e$ from $H'$ we notice that $\aRestrictionSet$ is not satisfied anymore, then it follows from Lemma~\ref{monotony} that there is no subset $H''$ of $H'$ that does not contain the edge $e$ such that $H''$ satisfies $\aRestrictionSet$ and such that $H''$ is a data-graph that still contains $\aGraph$. This allows us to assert that $e$ will belong to the final repair.

Considering all the previous discussions,  we design the algorithm \ref{algorithm:superset}. The first \textbf{for} will be executed at most $2^{\Sigma_n^\aRestrictionSet + 2}$ times. If $\Sigma_n$ is finite, then this number of executions is constant. Otherwise, it can depend exponentially on $|\aRestrictionSet|$. Everything inside the loop can be computed in polynomial time thanks to Theorem~\ref{teo:PGxpath}. The second loop will run at most $(|V_\aGraph| + \Sigma_n^\aRestrictionSet)^2$ times, which is quadratic in respect to the input size. Since the procedure is correct, we obtain the final result:

\begin{algorithm}[h!]
\begin{algorithmic}[1]
\REQUIRE $\aGraph$ is a data-graph and $\aRestrictionSet$ a set of $\Gposregxpath$ expressions.
\FOR{$S \in \mathcal{P}(\Sigma_n^\aRestrictionSet \cup \{c,d\} \setminus \Sigma_n^\aGraph)$}
    \STATE $H \leftarrow buildGraph(\aGraph, S)$
    \IF{$H \models \aRestrictionSet$ and $H \subseteq H'$}
        \STATE $H' \leftarrow H$
    \ENDIF
\ENDFOR
\IF{$H'$ is not initialized}
    \RETURN `There is no superset repair'
\ENDIF
\FOR{$e \in E_{H'} \setminus E_{\aGraph}$}
    \IF{$(V_{H'}, E_{H'} \setminus \{e\}, D_{H'}) \models \aRestrictionSet$}
        \STATE $H' \leftarrow (V_{H'}, E_{H'} \setminus \{e\}, D_{H'})$
    \ENDIF
\ENDFOR
\RETURN $H'$
\end{algorithmic}
\caption{$SupersetRepair(\aGraph,\aRestrictionSet)$}
\label{algorithm:superset}
\end{algorithm}

\begin{theorem}\label{teo:tractableSuperset}
Given a data-graph $\aGraph$ and a set of \Gposregxpath expressions, there is an algorithm that computes a superset repair of $\aGraph$ with respect to $\aRestrictionSet$ in polynomial time, if $\Sigma_n$ is finite.

If $\Sigma_n$ is infinite, then we can still compute a superset repair in polynomial time if $\aRestrictionSet$ is fixed.
\end{theorem}

Note that if $\aRestrictionSet$ only contains positive node expressions then we can use a similar procedure as in algorithm \ref{algorithm:subset}: we start by building the data-graph $H=buildGraph(G,S)$ where $S = \Sigma_n^\aRestrictionSet \cup \{c,d\} \setminus \Sigma_n^\aGraph$ and then we iteratively remove vertices that do not satisfy some node expression $\phi \in \aRestrictionSet$. By theorem \ref{teo:superset} we know that if there exists a superset repair, then there must exist some superset repair $G'$ such that $G \subseteq G' \subseteq H$. Also, if $v \notin \semantics{\phi}_H$ for some $v$ then $v \notin G'$. Therefore, in at most $|V_H \setminus V_\aGraph| \leq |\aRestrictionSet| + 2$ iterations we will find a data-graph that contains the same set of nodes as some superset repair of $\aGraph$ with respect to $\aRestrictionSet$. Then, we only need to remove the edges to reach the minimality condition. This yields the following result:

\begin{theorem}\label{teo:supersetPostNodeEpxressionTractable}
The $\exists$\textsc{SUPERSET-REPAIR} problem can be solved in polynomial-time in combined complexity if $\aRestrictionSet$ only contains node expressions from $\Gposregxpath$.
\end{theorem}

If we do not impose any of these restrictions the problem is once again intractable in general:

\begin{theorem}\label{teo:supsetHardInfinite}
If the set of \Gposregxpath expressions $\aRestrictionSet$ is not fixed and $\Sigma_n$ is allowed to be infinite, then the problem $\exists$\textsc{SUPERSET-REPAIR} is \textsc{NP-complete} in general.
\end{theorem}

\begin{proof}

    We reduce $3$\textsc{-SAT} to an instance of $\exists$\textsc{SUPERSET-REPAIR}, where $|R|$ depends on the input formula of $3$\textsc{-SAT} and $\Sigma_n$ is infinite. We take $\Sigma_e \supseteq \{ \esLabel{down}$\} and $\Sigma_n \supseteq \{x_i$ $|$ $i\in \N$\} $\cup$ \{$\lnot x_i$ $|$ $i \in \N\}$ $\cup$ $\{null\}$.
    
    Given the \textsc{3-CNF} formula $\phi$ with $n$ variables and $m$ clauses we define the data-graph $\aGraph=(V,L,D)$ where:
\begin{align*}
V_\aGraph &= \{v\}& \\
L(v,v) &= \emptyset& \\
D(v) &= null&
\end{align*}
    
    We would like that every superset repair of $\aGraph$ represents a valuation (or at least a partial valuation) on the variables $x_i$. Intuitively, given a superset repair $\aGraph'$, a variable $x_i$ will evaluate to $\top$ if there is a node in $\aGraph'$ with data value $x_i$, and will evaluate to $\bot$ if there is a node in $\aGraph'$ with data value $\lnot x_i$.
    
    In order to obtain valid valuations, we need to avoid the scenario where both data values $x_i$ and $\lnot x_i$ are present in a repair. To do this, we define the following constraints:
    
    \begin{equation}\label{path:validEval}
        \aPath_i = ([x_i^=] \down_\esLabel{down}^* [\lnot x_i^{\neq}]) \cup ([x_i^{\neq}] \down_\esLabel{down}^* [\lnot x_i^=]) \cup ([x_i^{\neq}] \down_\esLabel{down}^* [\lnot x_i^{\neq}])
    \end{equation}
    
    where $i$ iterates over $1 \leq i \leq n$. The expressions defined by $\ref{path:validEval}$ cannot be satisfied by any graph that has both nodes with data values $x_i$ and $\lnot x_i$ for some $i$. This follows from the fact that, for any of the three subexpressions, the path cannot begin and end in nodes having data values $x_i$ and $\lnot x_i$ respectively.
    
    Then, we add the following constraint to ensure that, for every repair, the defined valuation evaluates the clauses to true:
    
    \begin{equation}\label{path:satis}
        \aPathb_j = (\down_\esLabel{down}^* [(c_1^j)^=] \down_\esLabel{down}^*) \cup( \down_\esLabel{down}^* [(c_2^j)^=] \down_\esLabel{down}^* )\cup( \down_\esLabel{down}^* [(c_3^j)^=] \down_\esLabel{down}^*)
    \end{equation}
    
    where $j$ iterates over $1 \leq j \leq m$ and $c_k^j$ represents the $k$th literal appearing in $c_j$, either $x_i$ or $\lnot x_i$, for some $i$. Any repair of $\aGraph$ will satisfy all the expressions from~\ref{path:satis}, which implies that in any repair there must be at least one variable from every clause that evaluates to true.
    
    In summary, we defined $\aRestrictionSet = \{\aPath_i$ $|$ $1 \leq i \leq n$\} $\cup$ $\{\aPathb_j$ $|$ $1 \leq j \leq m\}$. Now we show that $\phi$ is satisfiable if and only if $\aGraph$ has a superset repair with respect to $\aRestrictionSet$.
    
    $\implies)$ Let $f$ be a valuation of the variables of $\phi$ that evaluates $\phi$ to true. We define the data-graph $H=(V_H,L_H,D_H)$ as:
    
\begin{align*}
V_H &= \{v\} \cup \{x_ \mid 1 \leq i \leq n \mbox{ and } f(x_i) = \top\} \cup \{\lnot x_i \mid 1 \leq i \leq n \mbox{ and } f(x_i)= \bot\}\\
L(a,b) &= \Sigma_e, \mbox{ for every pair } a,b \in V_H\\
D(v) &= null\\
D(x_i) &= x_i, \mbox{ for every node } x_i \in V_H\\
D(\lnot x_i) &= \lnot x_i, \mbox{ for every node }\lnot x_i \in V_H
\end{align*}
    
    Every expression $\aPath_i$ is satisfied, since there are no pairs of nodes with `opposite' data values and the graph is fully connected. The expressions $\aPathb_j$ are also satisfied given that for every clause $c_j$, at least one of the variables that evaluates the clause to true is present in $H$. Since $\aGraph\subseteq H$ and $H \models \aRestrictionSet$ there exists a superset repair $\aGraph'$ of $\aGraph$ with respect to $\aRestrictionSet$.

    $\impliedby)$ Let $\aGraph'$ be a superset repair of $\aGraph$ with respect to $\aRestrictionSet$. We define a valuation on the variables of $\phi$ by considering the data values present in $\aGraph'$. If the data value $x_i$ is present then $f(x_i) = \top$, and we define $f(x_i) = \bot$ otherwise. Observe that since $\aPath_i$ is satisfied for every $i$ this is a valid valuation. Moreover, notice that since $\aPathb_j$ is satisfied for every $j$, this valuation makes every clause evaluate to true. Let us consider an arbitrary clause $c_j$. Since $\aPathb_j$ is satisfied in $\aGraph'$, one of the literals $x_k^j$ for $k \in \{1,2,3\}$ has to appear as data value in $\aGraph'$. If one of those literals that appear is positive (i.e. without negation) then it evaluates to true by $f$, and thus $c_j$ is satisfied. Otherwise, there is a literal with negation from $c_j$ as data value in $\aGraph'$, and it evaluates to true through $f$, and thus $c_j$ evaluates to true. This shows that $\phi$ is satisfiable.
    
\end{proof}


\section{Related Work}\label{Section:RelatedWork}

Models for graph databases and knowledge graphs have been developed intensively since the 1990s, along with query languages and integrity constraints for them \cite{angles2008survey,barcelo2013querying}. During the last years they have earned significant attention from industry and academy due to their 
efficiency when modeling diverse, dynamic and large-scale collections of data \cite{hogan2020knowledge}. There are open source knowledge graphs such as YAGO \cite{rebele2016yago} or DBPEDIA \cite{lehmann2015dbpedia} that implement many features studied by the research community.

\Gregxpath is a query language developed for graph databases with data values in the nodes \cite{libkin2016querying} largely inspired by \textit{XPath} \cite{benedikt2009xpath}, a language for traversing data trees (i.e. XML documents). \Gregxpath expressiveness is quite understood in relation to other query languages \cite{libkin2016querying}, but there are still some open problems related to it, such as query containment and query equivalence.
Two other candidates for this work were the RQMs and RDQs ~\cite{libkin2012regular}, the first is more expressive than \Gregxpath, but its evaluation belongs to \textsc{PSpace}, while the other one is less expressive than \Gregxpath and therefore its complexity is lower.
The main advantage of \Gregxpath in comparison with other query languages is its great balance between expressiveness and complexity of evaluation (being polynomial on the graph and query size), as well as the fact that it can interact with the topology of the graph and at the same time reason about the data that it contains.

Different types of integrity and path constraints have been defined and studied in this context \cite{buneman2000path,abiteboul1999regular,barcelo2017data}. There is not yet a standard definition of \textit{consistency} for a graph database over some constraints, and therefore we developed our own based on a set of typical examples found in the literature. Other type of graph database constraints can be expressed through graph patterns and graph dependencies \cite{fan2019dependencies}. 

The notion of \textit{database repair} and the \textit{repair computing} problem, as well as the \textit{Consistent Query Answering} (CQA) problem, were first introduced in the relational context \cite{arenas1999consistent}. Since then CQA has received much attention from the research community in data management, developing techniques and efficient algorithms under different notions of repairs and considering all kinds of combinations of classes of integrity constraints and queries \cite{bertossi2011database,bertossi2019database}. These concepts were successfully extended to diverse data models such as \textit{XML} or \textit{description logics}  \cite{bertossi2006consistent,tenCate:2012,lukasiewicz2015classical,wijsen2005database}. In the case of graph databases there has been some work concerning CQA over graphs without data values \cite{barcelo2017data}, but, to the best of our knowledge, no previous work studies the problem of computing repairs neither the case when the graphs have data values in their nodes.

When studying CQA and repair problems in contexts such as description logics, it seems more natural to focus on subset repairs, as the DL semantics is \textit{open world} in nature. This means that the non presence of a fact in the database (explicitly or implicitly) is not enough to derive the negation of the fact. In the case of graph databases, some applications can respond  either to \textit{closed} or \textit{open world} semantics, and therefore both types of set-based repairs could be meaningful. In this work, we have considered both types of set-based repairs, but we note that we left for future work the study of \textit{symmetric difference}-based repairs, which were actually studied in~\cite{barcelo2017data}. Usually reasoning on symmetric difference repairs is much harder, and in the case of graphs with data the definition of \textit{symmetric difference} is not as direct to formalize as it is in the models defined in~\cite{barcelo2017data}.








\section{Conclusions}\label{Section:Conclusions}

In this work, we presented a graph database model along with a notion of consistency based upon path constraints defined by the \Gregxpath language. We proved many results concerning the complexity of computing subset and superset repairs of data-graphs given a set of restrictions (which can be summarized in Table~\ref{table:1}). We proved that, depending on the semantics of the repair and the syntactic conditions imposed on the set of constraints, the problem ranges from polynomial-time solvable to undecidable.

When restricting the language to the positive fragment of \Gregxpath, which we denote \Gposregxpath, we obtained polynomial-time algorithms (on data complexity) for superset-repair computing, and when considering only node expressions from \Gposregxpath we found a polynomial-time algorithm for computing subset repair. We note that the algorithm for subset repairs runs in polynomial time on both the size of the graph $\aGraph$ and the size of the restriction set $\aRestrictionSet$ (in other words, we consider \textit{combined complexity} of the problem). Furthermore, for every intractable case studied, we also obtained tight bounds on the combined and data complexity of the problem.

Some questions regarding repair computing in this context remain open. For example, if the $\exists$\textsc{SUPERSET-REPAIR} problem remains undecidable when considering only a fixed set of node expressions from \Gregxpath. It does not seem direct to extend the same proofs and techniques used for the case of path expressions or node expressions (under unfixed constraints).

Another venue of future research we are interested in is the development of preference criteria over repairs. In the problems we studied, we were concerned in finding \textit{any} repair of a data-graph given a set of constraints. This assumes no prior information about the domain semantics. However, in many contexts we might have a preference criterion that may yield an ordering over all possible repairs.

The complexity of CQA under this context has not been studied yet. Taking into account the results from \cite{Cate:2015} and \cite{barcelo2017data} it is plausible that reasoning over superset repairs will turn out to be harder than in subset repairs, but meanwhile the complexity of the problem remains unknown. Nevertheless, we note that the CQA problem turns out to be trivial in some cases based on the facts developed through this work concerning \Gregxpath and \Gposregxpath; for example, when considering subset repairs and \Gposregxpath there is a unique repair which is computable in polynomial time, and therefore the CQA problem is easy to solve.

\begin{table}[H]
\begin{tabular}{|l|p{6.2cm}|p{5.375cm}|}
\hline
 & \multicolumn{1}{c|}{Subset} & \multicolumn{1}{c|}{Superset}  \\ \hline
\Gregxpath   & \begin{tabular}[c]{@{}l@{}}NP-complete, even for a fixed set \\ of node constraints  (Th.~\ref{teo:subsetNodeHardness})\end{tabular} & \begin{tabular}[c]{@{}l@{}}Undecidable for fixed set of \\ path expressions (Th.~\ref{teo:supersetIndecidible}) and \\ for an unfixed set of node \\
expressions (Th.~\ref{teo:undecidabilityNodeExpressions})\end{tabular} \\
\hline
\multicolumn{1}{|c|}{\Gposregxpath} & \begin{tabular}[c]{@{}l@{}}PTime for node constraints. (Th.~\ref{teo:uniqueSubsetRepair}) \\ NP-complete in data complexity for \\ path constraints (Th.~\ref{teo:subsetPositiveIntractable}) \end{tabular} & \begin{tabular}[c]{@{}l@{}}PTime for any fixed constraints\\ (Th.~\ref{teo:tractableSuperset}) or for an unfixed set of \\ node expressions (Th.~\ref{teo:supersetPostNodeEpxressionTractable}). \\ NP-complete otherwise (Th.~\ref{teo:supsetHardInfinite}) \end{tabular} \\ \hline
\end{tabular}
\caption{Results for subset and superset repair computing problems.}
\label{table:1}
\end{table}

\section*{Acknowledgements}

We thank Ricardo Oscar Rodriguez (Departamento de Computación, Universidad de Buenos Aires) for his helpful comments and suggestions during the development of this work.



\section{Appendix}\label{Appendix}

\appendix

\begin{proof}[Proof of Lemma~\ref{monotony}]\label{Proof:Monotony}
We will prove it by induction over the structure of the expressions.
Let us note $\aGraph = (V_\aGraph, L_e, \dataFunction)$ and $\aGraph'= (V_{\aGraph'}, L_{e'}, \dataFunction')$.

In the base case of a path expression $\aPath$, it must be of the form $\epsilon$, $\_$, $\aLabel$ or $\aLabel^-$. In all of these cases the property holds: since $\aGraph \subseteq \aGraph'$, we have that $\{(x,x) \mid x \in V_{\aGraph}\} \subseteq \{(x,x) \mid x \in V_{\aGraph'}\}$, that  $\{(x,y) \mid x,y \in V_{\aGraph} \land L_e(x,y) \neq \emptyset \} \subseteq \{(x,y) \mid x,y \in V_{\aGraph'} \land L_{e'}(x,y) \neq \emptyset\}$, etc.

For the base case of a node expression $\aFormula$, it can only be of the form $\esDatoIgual{c}$ or $\esDatoDistinto{c}$, and it is easy to see that the property also holds.

For the inductive step, we consider all possible cases:

\begin{itemize}[label=$\circ$] 
    \item If $\aPath = \expNodoEnCamino{\aFormulab}$ then $\semantics{\aFormulab}_{\aGraph} \subseteq \semantics{\aFormulab}_{\aGraph'}$ (which holds by inductive hypothesis) implies $\semantics{\expNodoEnCamino{\aFormulab}}_{\aGraph} \subseteq \semantics{\expNodoEnCamino{\aFormulab}}_{\aGraph'}$.
    
    \item If $\aPath = \aPathb_1 . \aPathb_2$ then $(v,w) \in \semantics{\aPath}_{\aGraph}$ $\iff$ $\exists z$ such that $(v,z) \in \semantics{\aPathb_1}_{\aGraph}$ and $(z,w) \in \semantics{\aPathb_2}_{\aGraph}$. Since $\semantics{\aPathb_i}_{\aGraph} \subseteq \semantics{\aPathb_i}_{\aGraph'}$ for $i \in \{1,2\}$, then $(v,z) \in \semantics{\aPathb_1}_{\aGraph'}$ and $(z,w) \in \semantics{\aPathb_2}_{\aGraph'}$, which implies $(v,w) \in \semantics{\aPath}_{\aGraph'}$.
    
    \item If $\aPath = \aPathb_1 \cup \aPathb_2$ then $\semantics{\aPathb_i}_{\aGraph} \subseteq \semantics{\aPathb_i}_{\aGraph'}$ for $i \in \{1,2\}$ readily implies $\semantics{\aPath}_{\aGraph} \subseteq \semantics{\aPath}_{\aGraph'}$. 
    
    \item If $\aPath = \aPathb ^*$ then $(v,w) \in \semantics{\aPathb^*}_{\aGraph}$ $\iff$ $\exists z_1,z_2,...,z_m$ such that $z_1=v$, $z_m=w$ and $(z_i,z_{i+1}) \in \semantics{\aPathb}_{\aGraph}$ for $1\leq i < m$. Since $\semantics{\aPathb}_{\aGraph} \subseteq \semantics{\aPathb}_{\aGraph'}$ then $(z_i,z_{i+1}) \in \semantics{\aPathb}_{\aGraph'}$ for $1 \leq i < m$, which implies $(v,w) \in \semantics{\aPathb^*}_{\aGraph'}$.
    
    \item If $\aPath = \aPathb^{n,m}$ then $(v,w) \in \semantics{\aPathb^{n,m}}_\aGraph \iff \exists z_1,...,z_k, n+1 \leq k \leq m+1$ such that $z_1=v$, $z_k=w$ and $(z_i, z_{i+1}) \in \semantics{\aPathb}_\aGraph$ for $1 \leq i \leq k-1$. By the inductive hypothesis then $(z_i,z_{i+1}) \in \semantics{\aPathb}_{\aGraph'}$, which implies $(z,w) \in \semantics{\aPath^{n,m}}_{\aGraph'}$.
    
    \item If $\aPath = \aPathb_1 \cap \aPathb_2$ then $\semantics{\aPathb_i}_\aGraph \subseteq \semantics{\aPathb_i}_{\aGraph'}$ for $i \in \{1,2\}$ implies $\semantics{\aPathb_1 \cap \aPathb_2}_\aGraph = \semantics{\aPathb_1}_\aGraph \cap \semantics{\aPathb_2}_\aGraph \subseteq \semantics{\aPathb_1}_{\aGraph'} \cap \semantics{\aPathb_2}_{\aGraph'} = \semantics{\aPathb_1 \cap \aPathb_2}_{\aGraph'}$.
    
    \item If $\aFormula = \aFormulab_1 \wedge \aFormulab_2$ then $\semantics{\aFormulab_i}_{\aGraph} \subseteq \semantics{\aFormulab_i}_{\aGraph'}$ for $i \in \{1,2\}$ readily implies $\semantics{\aFormulab_1 \wedge \aFormulab_2}_{\aGraph} = \semantics{\aFormulab_1} \cap \semantics{\aFormulab_2} \subseteq \semantics{\aFormulab_1}_{\aGraph'} \cap \semantics{\aFormulab_2}_{\aGraph'} = \semantics{\aFormulab_1 \wedge \aFormulab_2}_{\aGraph'}$
    
    \item The case $\aFormula = \aFormulab_1 \vee \aFormulab_2$ can be treated in the same way.
    
    \item If $\aFormula = \comparacionCaminos{\aPathb}$ then $\semantics{\aPathb}_{\aGraph} \subseteq \semantics{\aPathb}_{\aGraph'}$ implies $\semantics{\comparacionCaminos{\aPathb}}_{\aGraph} \subseteq \semantics{\comparacionCaminos{\aPathb}}_{\aGraph'}$
    
    \item If $\aFormula = \comparacionCaminos{\aPathb_1 \star \aPathb_2}$, with $\star \in \{=, \neq\}$, then $v \in \semantics{\aFormula}_{\aGraph}$ $\iff$ $\exists z_1,z_2 \in V_\aGraph$ such that $(v, z_i) \in \semantics{\aPathb_i}_{\aGraph}$ for $i \in \{1,2\}$ and $\dataFunction(z_1) \star \dataFunction(z_2)$. Since $\semantics{\aPathb_i}_{\aGraph} \subseteq \semantics{\aPathb_i}_{\aGraph'}$, then $(v,z_i) \in \semantics{\aPathb_i}_{\aGraph'}$. And since we also have that $\dataFunction(z_i) = \dataFunction'(z_i)$ (since the $z_i$ are in $V_\aGraph$), we obtain that $v \in \semantics{\comparacionCaminos{\aPathb_1 \star \aPathb_2}}_{\aGraph'}$, as we wanted.
    

\end{itemize}
\end{proof}

\begin{proof}[Proof of Lemma~\ref{teo:dataValuesInR}]
In order to prove this Lemma we need to first prove another property of \Gposregxpath expressions:

\begin{lemma}\label{lemma:completeAndValues}
Let $\aPath$ be a \Gposregxpath path expression, $\aFormula$ a \Gposregxpath node expression, $E \subseteq \Sigma_e$ a set of edge labels and $K$ a data-graph where for every pair of nodes $v,w \in V_K$ it is the case that $L(v,w) = E$. Then it holds that:

\begin{enumerate} 
    \item \label{state:one}If $(v,w) \in \semantics{\aPath}_K$, $v\neq w$, and $\dataFunction_K(w) \notin \Sigma_n^\aPath$, then we have that $(v,z) \in \semantics{\aPath}_K$ for every node $z\in V_K$ such that $\dataFunction_K(z) \notin \Sigma_n^\aPath$. 
    
    \item \label{state:two}If $(w,v) \in \semantics{\aPath}_K$ , $v\neq w$, and $\dataFunction_K(w) \notin \Sigma_n^\aPath$, then we have that $(z,v) \in \semantics{\aPath}_K$ for every node $z\in V_K$ such that $\dataFunction_K(z) \notin \Sigma_n^\aPath$. 
   
    \item \label{state:three}If $(v,v) \in \semantics{\aPath}_K$ and $\dataFunction_K(v) \notin \Sigma_n^\aPath$ then $(z,z) \in \semantics{\aPath}_K$ for every node $z\in V_K$ such that $\dataFunction_K(z) \notin \Sigma_n^\aPath$.
    
    \item \label{state:four}If $v \in \semantics{\aFormula}_K$ for some $v \in V_K$ such that $\dataFunction_K(v) \notin \Sigma_n^\aFormula$ then $z \in \semantics{\aFormula}_K$ for every node $z\in V_K$ that satisfies $\dataFunction_K(z) \notin \Sigma_n^\aFormula$.
\end{enumerate}
\end{lemma}

\begin{proof}[Proof of Lemma~\ref{lemma:completeAndValues}]
We present a proof of the four properties simultaneously by structural induction. For the base cases:

\begin{itemize}[label=$\circ$]
    
    \item $\aPath=\labelComodin$: if $E \neq \emptyset$ then every pair $v,z \in V_K$ is in $\semantics{\labelComodin}_K$, and this case follows trivially. Otherwise there are no edges in the graph, and there isn't any hypothesis that can be satisfied.
    
    \item $\aPath = \aLabel$ and $\aPath = \aLabel^-$: if there is a pair of nodes $v,w \in V_K$ such that $(v,w) \in \semantics{\aLabel}_K$ then $\aLabel \in E$, and every pair of nodes $x,y \in V_K$ satisfies $(x,y) \in \semantics{\aLabel}_K$. The same reasoning holds for $\aLabel^-$.
    
    \item $\aPath = \epsilon$: this case follows trivially, since the hypothesis of the statements \ref{state:one} and \ref{state:two} cannot be satisfied, and the consequent of statement \ref{state:three} is always satisfied.
    
    \item $\aFormula = \esDatoIgual{e}$: the hypothesis of statement \ref{state:four} cannot be satisfied in this case, and then this follows trivially.
    
    \item $\aFormula = \esDatoDistinto{e}$: every node $v \in V_K$ such that $\dataFunction_K(v) \notin \Sigma_n^{\esDatoDistinto{e}} = \{\esDato{e}\}$
    satisfies $v \in \semantics{\esDatoDistinto{e}}_K$.
    
\end{itemize}

We now proceed with the inductive cases:

\begin{itemize}[label=$\circ$]

    \item $\aPath = [\aFormulab]$: the hypothesis of statements~\ref{state:one} and \ref{state:two} cannot be satisfied in this case, and therefore we only consider statement~\ref{state:three}. Suppose there is a node $v \in V_K$ such that $\dataFunction_K(v) \notin \Sigma_n^ {[\aFormulab]}$ and $(v,v) \in \semantics{[\aFormulab]}_K$. By hypothesis, every other node $z \in V_K$ such that $\dataFunction_K(z) \notin \Sigma_n^{\aFormulab}$ satisfies $z \in \semantics{\aFormulab}_K$, which then implies that $(z,z) \in \semantics{[\aFormulab]}_K$.
    
    \item $\aPath = \aPathb_1 . \aPathb_2$: for the statement~\ref{state:one}, let $v,w \in V_K$ satisfy its hypothesis ($(v,w) \in \semantics{\aPath}_K$, $v\neq w$, $\dataFunction(w)\notin \Sigma_n^{\aPath}$), and let $x$ be a node such that $(v,x) \in \semantics{\aPathb_1}_K$ and $(x,w) \in \semantics{\aPathb_2}_K$. Then if $x\neq w$ we know by the inductive hypothesis that $(x,z) \in \semantics{\aPathb_2}_K$ for every $z\in V_K$ with $\dataFunction_K(z) \notin \Sigma_n^{\aPathb_2}$  (which is a superset of those nodes $y$ such that $\dataFunction_K(y) \notin \Sigma_n^{\aPathb_1 . \aPathb_2}$), and then we can conclude that $(v,z) \in \semantics{\aPathb_1 . \aPathb_2}_K$. If $w=x$ the result follows by considering that $(v,z) \in \semantics{\aPathb_1}_K$ and then $(z,z) \in \semantics{\aPathb_2}_K$ by the statement~\ref{state:three}. The statement~\ref{state:two} can be proved the same way.
    
    Now for statement~\ref{state:three} let $v$ be a node from $K$ that satisfies its hypothesis ($(v,v) \in \semantics{\aPath}_K$ and $\dataFunction_K(z) \notin \Sigma_n^{\aPath}$) and let $x$ be the `intermediate' node. Let $z$ be a node such that $\dataFunction_K(z) \notin \Sigma_n^{\aPathb_1 . \aPathb_2}$. If $x=v$ it follows by using statement~\ref{state:three} as inductive hypothesis that $(z,z) \in \semantics{\aPathb_1 . \aPathb_2}_K$. Otherwise (i.e. $x \neq v$) by using statement~\ref{state:one} and \ref{state:two} as inductive hypothesis we can obtain the same result.

    \item $\aPath = \aPathb_1 \cup \aPathb_2$: this follows by inductive hypothesis considering, without loss of generality, that $(v,w) \in \semantics{\aPathb_1}_K$ (and respectively $(w,v)$ and $(v,v)$ for statements $\ref{state:two}$ and $\ref{state:three}$).
    
    \item The case $\aPath = \aPathb_1 \cap \aPathb_2$ can be treated in a similar way.
    
    \item $\aPath = \aPathb^*$: if $(v,w) \in \semantics{\aPathb^*}_K$ then either $(v,w) \in \semantics{\epsilon}_K$ or $(v,w) \in \semantics{\aPathb^n}_K$ for some $n > 0$. The first scenario has already been proved, while the second one follows from what we have seen in the concatenation case.
    
    \item If $\aPath = \aPathb^{n,m}$: the result follows with the same argument used in the $\aPathb^*$ case since $(v,w) \in \semantics{\aPathb^{n,m}} \iff (v,w) \in \semantics{\aPathb^k}$ for some $n \leq k \leq m$.

    \item $\aFormula = \aFormulab_1 \wedge \aFormulab_2$: let $v \in V_K$ be a node such that $v \in \semantics{\aFormula}_K$ and $\dataFunction_K(v) \notin \Sigma_n^\aFormula$. By inductive hypothesis, every other $z \in V_K$ with $\dataFunction_K(z) \notin \aFormula$ satisfies $z \in \semantics{\aFormulab_i}_K$ for $i \in \{1,2\}$, which implies that $z \in \semantics{\aFormula}_K$.
    
    \item $\aFormula = \aFormulab_1 \vee \aFormulab_2$: let $v \in V_K$ satisfy the hypothesis of statement~\ref{state:four}, and without loss of generality assume $v \in \semantics{\aFormulab_1}_K$. By inductive hypothesis every other node $z \in V_K$ such that $D_K(z) \notin \Sigma_n^\aFormula$ satisfies $z \in \semantics{\aFormulab_1}_K$, which implies that $z \in \semantics{\aFormulab_1 \vee \aFormulab_2}_K$.
    
    \item $\aFormula = \comparacionCaminos{\aPathb}$: let $v \in V_K$ satisfy the hypothesis of statement~\ref{state:four}. Then $(v,x) \in \semantics{\aPathb}_K$ for some $x$. Thanks to statement~\ref{state:two} we can deduce that if $x \neq v$ then $(z,x) \in \semantics{\aPathb}_K$ for every $z$ such that $\dataFunction_K(z) \notin \Sigma_n^{\comparacionCaminos{\aPathb}}$. If $x=v$ we use statement~\ref{state:three}.
    
    \item $\aFormula = \comparacionCaminos{\aPathb_1 = \aPathb_2}$: let $v \in V_K$ satisfy the hypothesis of statement~\ref{state:four}, $x_1$ and $x_2$ be nodes such that $(v,x_i) \in \semantics{\aPathb_i}_K$ for $i \in \{1,2\}$ and  $\dataFunction_K(x_1) = \dataFunction_K(x_2)$, and $z \in V_K$ other node such that $\dataFunction_K(z) \notin \Sigma_n^\aFormula$. If both $x_i$ are different from $v$ then $(z,x_i) \in \semantics{\aPathb_i}_K$ thanks to statement $\ref{state:two}$. Otherwise, suppose $x_1 = v$ without loss of generality, which then implies that $(z,z) \in \semantics{\aPathb_1}_K$ due to statement $\ref{state:three}$. Since $\dataFunction_K(x_2) = \dataFunction_K(x_1) = \dataFunction_K(v) \notin \Sigma_n^\aFormula$ we can guarantee that $(z,z) \in \semantics{\aPathb_2}_K$: if $x_2 = v$ then this follows due to statement~\ref{state:three}, and if $x_2 \neq v$ then due to statement~\ref{state:two} we know that $(x_2,x_2) \in \semantics{\aPathb_2}_K$ and then that $(z,z) \in \semantics{\aPathb_2}_K$ through statement $\ref{state:three}$.

    \item $\aFormula = \comparacionCaminos{\aPathb_1 \neq \aPathb_2}$: let $v \in V_K$ satisfy the hypothesis of statement~\ref{state:four}, $x_1$ and $x_2$ be nodes such that $(v,x_i) \in \semantics{\aPathb_i}_K$ for $i \in \{1,2\}$ and  $\dataFunction_K(x_1) \neq \dataFunction_K(x_2)$, and $z \in V_K$ other node such that $\dataFunction_K(z) \notin \Sigma_n^\aFormula$. If $v \neq x_i$ for $i \in \{1,2\}$ then the result follows by statement~\ref{state:two}. Otherwise, consider that $x_1 = v$, which then implies that $(z,z) \in \semantics{\aPathb_1}_K$ by statement~\ref{state:three}. If $x_2 \neq v$ then by statement~\ref{state:two} we can conclude that $(z,x_2) \in \semantics{\aPathb_2}_K$. If $\dataFunction_K(x_2) \neq \dataFunction_K(z)$ we have finished, and if not, knowing that $\dataFunction_K(x_2) \notin \Sigma_n^\aFormula$ (this follows from $\dataFunction_K(x_2) = \dataFunction_K(z)$), we can consider that $(z,v) \in \semantics{\aPathb_2}_K$ by statement $\ref{state:one}$ and that $\dataFunction_K(x_2) \neq \dataFunction_K(x_1) = \dataFunction_K(v)$. The case $x_2 = v$ cannot happen, since we are assuming that $x_1 = v$ and $\dataFunction_K(x_1) \neq \dataFunction_K(x_2)$.
    
\end{itemize}

\end{proof}

Now we continue with the Lemma~\ref{teo:dataValuesInR}.

Let $\aGraph'$ be a superset repair of $\aGraph$ with respect to $\aRestrictionSet$. If $\Sigma_n^{\aGraph'} \subseteq \Sigma_n^\aGraph \cup \Sigma_n^\aRestrictionSet \cup \{c,d\}$ for some $c,d \in \Sigma_n \setminus \Sigma_n^\aRestrictionSet$, $c \neq d$, then $\aGraph'$ is already the witness we want for this lemma. Otherwise, take $c,d \in \Sigma_n^{\aGraph'} \setminus (\Sigma_n^\aGraph \cup \Sigma_n^\aRestrictionSet)$, $c\neq d$. We define the graph $H=(V_H,L_H,\dataFunction_H)$ where:

\begin{center}
    $V_H = V_{\aGraph'}$
    
    \bigskip 
    
    $L_H(v,w) = \Sigma_e^{G'} \quad \forall v,w \in V_H$
    
     \bigskip 
    
    
    
    
     $$\dataFunction_H(v)=\begin{cases}
		\dataFunction_{\aGraph'}(v)  &  \text{\small if $\dataFunction_{\aGraph'}(v) \in \Sigma_n^{\aGraph} \cup \Sigma_n^\aRestrictionSet \cup \{c, d\}$ } \\ 
		c   & \text{otherwise}
		\end{cases}
	   $$
    
\end{center}

where $\Sigma_e^{G'}$ is the set of all edge labels mentioned in $G'$ (i.e. $\Sigma_e^{G'} = \{l \in \Sigma_e$ $|$ $\exists v,w \in V_{G'}$ such that $l \in L_{G'}(v,w)\}$).

Intuitively, we have changed all data values that were not in $\Sigma_n^\aGraph \cup \Sigma_n^\aRestrictionSet \cup \{c,d\}$ into $c$, and we also added every possible edge considering those labels already present in $\aGraph'$. Note that $\Sigma_n^H \subseteq \Sigma_n^\aGraph \cup \Sigma_n^\aRestrictionSet \cup \{c,d\}$, so showing that $H \models \aRestrictionSet$ is enough to prove the lemma. We will show by induction in the expressions' structure that if a pair $v,w \in V_{\aGraph'}$ satisfies $(v,w) \in \semantics{\aPath}_{\aGraph'}$ then $(v,w) \in \semantics{\aPath}_H$, and that if $v \in \semantics{\aFormula}_{\aGraph'}$ then $v \in \semantics{\aFormula}_H$ (where the expressions considered only use data values from $\Sigma_n^\aRestrictionSet$). 

For the base cases:

\begin{itemize}[label=$\circ$]
    
    \item $\aPath=\labelComodin$: if some pair $(v,w)$ is in $\semantics{\labelComodin}_{\aGraph'}$ then $\Sigma_e^{G'} \neq \emptyset$, and $(v,w) \in \semantics{\labelComodin}_{\aGraph'}$.
    
    \item $\aPath = \aLabel$ and $\aPath = \aLabel^-$: if $(v,w) \in \semantics{\aLabel}_{G'}$ then clearly $\aLabel \in \Sigma_e^{G'}$ and $\aLabel \in L_{H}(v,w)$.
    
    \item $\aPath = \epsilon$: this case follows trivially.
    
    \item $\aFormula = \esDatoIgual{e}$: if $v \in \semantics{\esDatoIgual{e}}_{\aGraph'}$ then $\dataFunction_{\aGraph'}(v)=\esDato{e}$. Since $\esDato{e} \in \Sigma_n^\aRestrictionSet$ the data value was preserved in $H$ ($\dataFunction_H(v)=\esDato{e}$) and thus $v \in \semantics{\esDatoIgual{e}}_H$.
    
    \item $\aFormula = \esDatoDistinto{e}$: let $v \in \semantics{\esDatoDistinto{e}}_{\aGraph'}$. If $\dataFunction_{\aGraph'}(v) \in \Sigma_n^\aGraph \cup \Sigma_n^\aRestrictionSet \cup \{c, d\}$ then the data value of $v$ was preserved, and $v \in \semantics{\esDatoDistinto{e}}_H$. Otherwise $\dataFunction_H(v)= \esDato{c}$, and since $\esDato{c} \notin \Sigma_n^\aRestrictionSet$ we know $v \in \semantics{\esDatoDistinto{e}}_H$.
    
\end{itemize}

Now, for the inductive cases:

\begin{itemize}[label=$\circ$]

    \item $\aPath = [\aFormulab]$: let $(v,v) \in \semantics{[\aFormulab]}_{\aGraph'}$, which implies that $v \in \semantics{\aFormulab}_{\aGraph'}$. By inductive hypothesis we have that $v \in \semantics{\aFormulab}_H$, and then $(v,v) \in \semantics{[\aFormulab]}_H$.
    
    \item $\aPath = \aPathb_1 . \aPathb_2$: if $(v,w) \in \semantics{\aPathb_1 . \aPathb_2}_{\aGraph'}$ then there exists $z \in V_{\aGraph'}$ such that $(v,z) \in \semantics{\aPathb_1}_{\aGraph'}$ and $(z,w) \in \semantics{\aPathb_2}_{\aGraph'}$. By inductive hypothesis $(v,z) \in \semantics{\aPathb_1}_H$ and $(z,w) \in \semantics{\aPathb_2}_H$, which implies that $(v,w) \in \semantics{\aPathb_1 . \aPathb_2}_H$.

    \item $\aPath = \aPathb_1 \cup \aPathb_2$: let $(v,w) \in \semantics{\aPathb_1 \cup \aPathb_2}_{\aGraph'}$. Without loss of generality, we can assume that $(v,w) \in \semantics{\aPathb_1}_{\aGraph'}$. By inductive hypothesis $(v,w) \in \semantics{\aPathb_1}_H$, and then $(v,w) \in \semantics{\aPathb_1 \cup \aPathb_2}_H$.
    
    \item $\aPath = \aPathb_1 \cap \aPathb_2$: since $(v,w) \in \semantics{\aPathb_1 \cap \aPathb_2}_\aGraph \iff (v,w) \in \semantics{\aPathb_1}_\aGraph \cap \semantics{\aPathb_2}_\aGraph$ then the inductive hypothesis implies $(v,w) \in \semantics{\aPathb_1}_H \cap \semantics{\aPathb_2}_H = \semantics{\aPathb_1 \cap \aPathb_2}_H$.
    
    \item $\aPath = \aPathb^*$: $(v,w) \in \semantics{\aPathb^*}_{\aGraph'}$ if there exists $z_1,...z_m \in V_{\aGraph'}$ such that $z_1 = v$, $z_m = w$ and $(z_i, z_{i+1}) \in \semantics{\aPathb}_{\aGraph'}$ for $1\leq i \leq m-1$. By inductive hypothesis $(z_i,z_{i+1}) \in \semantics{\aPathb}_H$ for $1 \leq i \leq m-1$, and then $(v,w) \in \semantics{\aPathb^*}_H$.
    
    \item The case $\aPath = \aPathb^{n,m}$ follows with the same argument used in the $.^*$ case.
    
    \item $\aFormula = \aFormulab_1 \wedge \aFormulab_2$: let $v \in \semantics{\aFormulab_1 \wedge \aFormulab_2}_{\aGraph'}$. Since $v \in \semantics{\aFormulab_i}_{\aGraph'}$ for $i \in \{1,2\}$ we know by hypothesis that $v \in \semantics{\aFormulab_i}_H$ for $i \in \{1,2\}$, which implies that $v \in \semantics{\aFormulab_1 \wedge \aFormulab_2}_H$.
    
    \item $\aFormula = \aFormulab_1 \vee \aFormulab_2$: we can assume without loss of generality that if $v \in \semantics{\aFormulab_1 \vee \aFormulab_2}_\aGraph$ then $v \in \semantics{\aFormulab_1}_\aGraph$, and we can readily conclude by inductive hypothesis that $v \in \semantics{\aFormulab_1}_H \subseteq \semantics{\aFormulab_1 \vee \aFormulab_2}_H$.
    
    \item $\aFormula = \comparacionCaminos{\aPathb}$: if $v \in \semantics{\comparacionCaminos{\aPathb}}_{\aGraph'}$ then there exists $z\in V_{\aGraph'}$ such that $(v,z) \in \semantics{\aPathb}_{\aGraph'}$. Then by hypothesis $(v,z) \in \semantics{\aPathb}_H$ and $v \in \semantics{\comparacionCaminos{\aPathb}}_H$.
    
    \item $\aFormula = \comparacionCaminos{\aPathb_1 = \aPathb_2}$: if $v \in \semantics{\comparacionCaminos{\aPathb_1 = \aPathb_2}}_{\aGraph'}$ then there exists $z_1,z_2 \in V_{\aGraph'}$ such that $(v,z_1) \in \semantics{\aPathb_1}_{\aGraph'}$, $(v,z_2) \in \semantics{\aPathb_2}_{\aGraph'}$ and $\dataFunction_{\aGraph'}(z_1)=\dataFunction_{\aGraph'}(z_2)$. By construction $\dataFunction_H(z_1)=\dataFunction_H(z_2)$, and by hypothesis $(v,z_1) \in \semantics{\aPathb_1}_H$ and $(v,z_2) \in \semantics{\aPathb_2}_H$, which implies that $v \in \semantics{\comparacionCaminos{\aPathb_1 = \aPathb_2}}_H$. 
    
    \item $\aFormula = \comparacionCaminos{\aPathb_1 \neq \aPathb_2}$: this is the most interesting case. Note that the paths used in $\aGraph$ to satisfy the expression could now lead to nodes with the same data value. If $v \in \semantics{\comparacionCaminos{\aPathb_1 \neq \aPathb_2}}_{\aGraph'}$ then there exist nodes $z_1,z_2 \in V_{\aGraph'}$ such that $(v,z_1) \in \semantics{\aPathb_1}_{\aGraph'}$, $(v,z_2) \in \semantics{\aPathb_2}_{\aGraph'}$ and $\dataFunction_{\aGraph'}(z_1) \neq \dataFunction_{\aGraph'}(z_2)$. If $\dataFunction_H(z_1) \neq \dataFunction_H(z_2)$ then we can conclude using the inductive hypothesis that $v \in \semantics{\comparacionCaminos{\aPathb_1 \neq \aPathb_2}}_H$. If this is not the case, then both nodes had in ${\aGraph'}$ as data value either $c$ or one of those we changed (i.e. one from $\Sigma_n^{\aGraph'} \setminus (\Sigma_n^\aGraph \cup \Sigma_n^\aRestrictionSet \cup \{c,d\})$. Note that since $\dataFunction_{\aGraph'}(z_1) \neq \dataFunction_{\aGraph'}(z_2)$, it must hold that either $z_1 \neq v$ or $z_2 \neq v$; without loss of generality we can assume that $z_1 \neq v$. Let $v_d$ be a node from $H$ that has data value $d$ (by construction there must be at least one such node). Since $(v,z_1) \in \semantics{\aPathb_1}_H$, $v\neq z_1$, $\dataFunction_H(z_1) \notin \Sigma_n^\aRestrictionSet$ and $H$ is a `complete graph' under edge labels form $\Sigma_e^{\aGraph'}$ we can deduce using statement \ref{state:one} of Lemma~\ref{lemma:completeAndValues} that $(v,v_d) \in \semantics{\aPathb_1}_H$. Then we can conclude that $v \in \semantics{\comparacionCaminos{\aPathb_1 \neq \aPathb_2}}_H$ taking into account that $(v,z_2) \in \semantics{\aPathb_2}_H$ due to the inductive hypothesis.
    
\end{itemize}

We can then conclude, since $\aGraph' \models \aRestrictionSet$, that $H \models \aRestrictionSet$ holds.

\end{proof}

\begin{proof}[Proof of Lemma~\ref{teoremaHorrible}]
Let $f:V_\aGraph \rightarrow V_H$ be a mapping between the nodes of $\aGraph$ and $H$ such that $f(v) = v$ if $v \in V_\aGraph \setminus V_d$, and $f(v)=v_d$ otherwise. Now we show, by induction over the formula's structure, that if a pair of nodes $v,w \in V_\aGraph$ satisfies $(v,w) \in \semantics{\aPath}_\aGraph$, then $(f(v),f(w)) \in \semantics{\aPath}_H$, and that if $v \in \semantics{\aFormula}_\aGraph$ then $f(v) \in \semantics{\aFormula}_H$. This is enough to prove the lemma.

For the base cases:

\begin{itemize} [label=$\circ$]

    \item $\aPath = \labelComodin$: if both $v$ and $w$ were not in $V_d$ then their edges were preserved. If $w \in V_d$ and $v \notin V_d$ then by construction $L_\aGraph(v,w) \subseteq L_H(v,v_d)$. The case when $v \in V_d$ is analogous. If both $v$ and $w$ are in $V_d$ we use the fact that $L_\aGraph(v,w) \subseteq L_H(v_d,v_d)$.
    
    \item $\aPath = \aLabel$ or $\aPath = \aLabel^-$: idem the previous case.
    
    \item $\aPath = \epsilon$: this one follows trivially.
    
    \item $\aFormula = \esDatoIgual{c}$: $v$ and $f(v)$ have the same data value, which means that this case holds.
    
    \item $\aFormula = \esDatoDistinto{c}$: idem the previous case.
    
\end{itemize}

Now, the inductive ones:

\begin{itemize}[label=$\circ$]
    
    \item $\aPath = [\aFormulab]$: If $(v,v) \in \semantics{[\aFormulab]}_\aGraph$ then $v \in \semantics{\aFormulab}_\aGraph$, which implies by inductive hypothesis that $f(v) \in \semantics{\aFormulab}_H$ and then $(f(v),f(v)) \in \semantics{[\aFormulab]}_H$.
    
    \item $\aPath = \aPathb_1 . \aPathb_2$: if $(v,w) \in \semantics{\aPathb_1 . \aPathb_2}_\aGraph$ then there exists $z \in V_\aGraph$ such that $(v,z) \in \semantics{\aPathb_1}_\aGraph$ and $(z,w) \in \semantics{\aPathb_2}_\aGraph$. Then by inductive hypothesis $(f(v),f(z)) \in \semantics{\aPathb_1}_H$ and $(f(z), f(w)) \in \semantics{\aPathb_2}_H$, which implies that $(f(v),f(w)) \in \semantics{\aPathb_1 . \aPathb_2}_H$.
    
    \item $\aPath = \aPathb_1 \cup \aPathb_2$: suppose without loss of generality that $(v,w) \in \semantics{\aPathb_1}_\aGraph$. Then by induction $(f(v),f(w)) \in \semantics{\aPathb_1 \cup \aPathb_2}_H$.
    
    \item $\aPath = \aPathb_1 \cap \aPathb_2$: if $(v,w) \in \semantics{\aPathb_1 \cap \aPathb_2}_\aGraph$ then $(v,w) \in \semantics{\aPathb_1}_\aGraph \cap \semantics{\aPathb_2}_\aGraph$, and by the inductive hypothesis we can deduce that $(f(v),f(w)) \in \semantics{\aPathb_1}_H \cap \semantics{\aPathb_2}_H = \semantics{\aPathb_1 \cap \aPathb_2}_H$.
    
    \item $\aPath = \aPathb^*$: $(v,w) \in \semantics{\aPathb^*}_\aGraph$ if there exists $z_1,...,z_m \in V_\aGraph$ such that $z_1=v$, $z_m = w$ and $(z_i,z_{i+1}) \in \semantics{\aPathb}_\aGraph$ for $1\leq i \leq m-1$. By hypothesis it follows that $(f(z_i),f(z_{i+1})) \in \semantics{\aPathb}_H$, which implies that $(f(v),f(w)) \in \semantics{\aPathb^*}_H$.
    
    \item $\aPath = \aPathb^{n,m}$: this case follows using the same argument used in the $.^*$ case.
    
    \item $\aFormula = \aFormulab_1 \wedge \aFormulab_2$: $v \in \semantics{\aFormulab_1 \wedge \aFormulab_2}_\aGraph$ if $v \in \semantics{\aFormulab_1}_\aGraph$ and $v \in \semantics{\aFormulab_2}_\aGraph$. By using the hypothesis we conclude that $f(v) \in \semantics{\aFormulab_1}_H$ and $f(v) \in \semantics{\aFormulab_2}_H$, which implies that $f(v) \in \semantics{\aFormulab_1 \wedge \aFormulab_2}_H$.
    
    \item $\aFormula = \aFormulab_1 \vee \aFormulab_2$: if $v \in \semantics{\aFormulab_1 \vee \aFormulab_2}_\aGraph$ then without loss of generality we can assume $v \in \semantics{\aFormulab_1}_\aGraph$. Then by inductive hypothesis $f(v) \in \semantics{\aFormulab_1}_H \subseteq \semantics{\aFormulab_1 \vee \aFormulab_2}_H$.
    
    \item $\aFormula = \comparacionCaminos{\aPathb}$: $v \in \semantics{\comparacionCaminos{\aPathb}}_\aGraph$ if $(v,z) \in \semantics{\aPathb}_\aGraph$ for some $z$. By hypothesis $(f(v),f(z)) \in \semantics{\aPathb}_H$, and then $f(v) \in \semantics{\comparacionCaminos{\aPathb}}_H$. 
    
    \item $\aFormula = \comparacionCaminos{\aPathb_1 \star \aPathb_2}$, with $\star \in \{=, \neq\}$: $v \in \semantics{\comparacionCaminos{\aPathb_1 \star \aPathb_2}}$ if there exists $z_1,z_2 \in V_\aGraph$ such that $(v,z_1) \in \semantics{\aPathb_1}_\aGraph$, $(v,z_2) \in \semantics{\aPathb_2}_\aGraph$ and $\dataFunction_\aGraph(z_1) \star \dataFunction_\aGraph(z_2)$. In $H$ both $z_1$ and $z_2$ preserve their data values (i.e. $\dataFunction_\aGraph(z_i) = \dataFunction_H(f(z_i))$ for $i \in \{1,2\}$), and by inductive hypothesis $(f(v),f(z_i)) \in \semantics{\aPathb_i}_H$ for $i \in \{1,2\}$. Then $f(v) \in \semantics{\comparacionCaminos{\aPathb_1 \star \aPathb_2}}_H$.
    
    
\end{itemize}
\end{proof}

\vskip 0.2in
\bibliography{sample}
\bibliographystyle{theapa}
\end{document}